\documentclass{article} % For LaTeX2e
\usepackage{iclr2027_conference,times}

\usepackage{amsmath,amsfonts,bm}

\def\eqref#1{equation~\ref{#1}}
\def\1{\bm{1}}

\def\rva{{\mathbf{a}}}

\def\rvu{{\mathbf{i}}}

\def\rvm{{\mathbf{m}}}

\def\rvs{{\mathbf{s}}}

\def\rvu{{\mathbf{u}}}

\def\rvw{{\mathbf{w}}}

\DeclareMathAlphabet{\mathsfit}{\encodingdefault}{\sfdefault}{m}{sl}
\SetMathAlphabet{\mathsfit}{bold}{\encodingdefault}{\sfdefault}{bx}{n}

\DeclareMathOperator*{\argmax}{arg\,max}

\usepackage{hyperref}
\usepackage{url}

\usepackage{algorithm}
\usepackage{amsmath,amssymb,amsfonts}
\usepackage{amsthm}
\usepackage{algorithm}
\usepackage{algorithmic}
\usepackage[table,xcdraw]{xcolor}
\usepackage{graphicx}
\usepackage{textcomp}
\usepackage{xcolor}
\usepackage{booktabs}
\usepackage{multirow}
\usepackage{subcaption}
\usepackage{wrapfig}

\newcommand{\set}{\mathcal}

\newtheorem{theorem}{\textbf{Theorem}}

\newtheorem{assumption}{\textbf{Assumption}}
\newtheorem{lemma}{\textbf{Lemma}}

\title{Embedding Information in Policy Randomness: Distribution-Preserving Implicit \\Communication in MDPs}

\author{Gongpu Chen, Deniz Gündüz \\
Department of Electrical and Electronic Engineering\\
Imperial College London,\\
London SW7 2AZ, U.K. \\
\texttt{\{gongpu.chen,d.gunduz\}@imperial.ac.uk} \\
}

\iclrfinalcopy % Uncomment for camera-ready version, but NOT for submission.
\begin{document}

\maketitle

\begin{abstract}
This paper investigates a new role of randomness in stochastic policies---as a medium for embedding and communicating information. Specifically, we consider an agent interacting with an environment according to a prescribed stochastic policy. By carefully modulating the random action realizations of the policy, information can be embedded into the agent's behavior without altering the prescribed action distribution. We cast this problem as a channel coding problem, in which the controlled system implicitly serves as a communication channel. We then propose a coding scheme, termed posterior-matching randomness modulation (PRAM), that selects actions according to the message to be transmitted. With the aid of common randomness independent of the control system, PRAM preserves the prescribed action distribution for every state and every message. We further prove that PRAM is asymptotically reliable: the message can be decoded with arbitrarily small error probability as the observed state trajectory becomes sufficiently long. Experiments demonstrate the effectiveness of the proposed framework in applications including multi-bit policy watermarking and trajectory provenance.
\end{abstract}

\section{Introduction}
Stochastic policies are widely used in control systems and reinforcement learning (RL), where policy randomness facilitates exploration~\citep{ahmed2019understanding}, enhance robustness~\citep{eysenbach2021maximum}, and supports multimodal decision making~\citep{zhao2019maximum}. Under a stochastic policy, actions are sampled from a state-dependent distribution rather than selected deterministically. This paper reveals a new role of policy randomness: it can be exploited to embed information into an agent's behavior without degrading its control performance. Specifically, we consider a Markov decision process (MDP) in which an agent interacts with the environment according to a prescribed stochastic policy~\citep{puterman2014markov}. Our goal is to enable the agent to embed and communicate a message through its action realizations while preserving the prescribed action distribution at every state. In other words, the embedded information should be recoverable from the resulting state trajectory without changing the marginal behavior of the reference policy. We refer to this technique as \textit{distribution-preserving implicit communication (DPIC)}.

One important application of this technique is multi-bit policy watermarking. RL has been increasingly deployed in practical decision-making systems, ranging from robotics and autonomous systems to energy management, inventory control, and cyber-defence~\citep{zifan2025reinforcement,yu2021review,nguyen2021deep}. In many such applications, the resulting policies are trained using valuable data and substantial computational and engineering resources, making them important intellectual property (IP) of the development team. Protecting the IP of these RL policies therefore requires mechanisms that can verify policy ownership and detect unauthorized use or replication. Motivated by this need, recent studies \citep{amanda2026remotely,chen2021temporal,behzadan2019sequential,huang2025agent} have investigated digital watermarking methods for policy IP protection. However, existing methods primarily support zero-bit watermarking, where the detector only determines whether a watermark is present. In contrast, the DPIC framework naturally enables multi-bit watermarking, allowing policy owners to embed richer information for identification and attribution.
The same information-embedding mechanism can also support closely related tasks such as trajectory provenance~\citep{gao2026trace} and more general behavioral metadata embedding~\citep{an2026sequential} for source attribution, authentication, and traceability.

As suggested by its name, DPIC casts information embedding as a communication problem, and more specifically, a channel coding problem. Inspired by implicit communication in control systems~\citep{chen2025implicit,Act2Comm}, we view the MDP environment as a state-dependent communication channel, where the agent's action serves as the channel input and the resulting state serves as the channel output. Under this interpretation, embedding a $B$-bit message into the agent's behavior is equivalent to transmitting the message through this channel. 

A key objective of DPIC is to enable such communication without degrading the original control performance. This motivates a distribution-preserving requirement on the agent's actions. Unlike conventional communication, where the input distribution is typically constrained only after averaging over the message distribution~\citep{cover1991elements}, applications such as policy watermarking may repeatedly embed the same fixed message. Message-averaged preservation is therefore insufficient to guarantee unchanged control performance. We instead require the action distribution to match the prescribed stochastic policy for every state and \textit{every message}. 
This requirement poses a fundamental challenge. From an information-theoretic perspective, preserving the same action distribution for every message implies that, conditioned on the current state, the action alone is independent of the message and therefore cannot directly convey information about it. To overcome this limitation, we introduce common randomness shared between the agent (encoder) and the observer (decoder), and select each action based on the message, the observed state history, and the shared random variable. Communication is then enabled through the message-dependent coupling between the action and the common randomness, while the marginal action distribution remains unchanged for every message.

Building on this idea, we develop the \textit{posterior-matching randomness modulation (PRAM)} scheme. The term \textit{randomness modulation} refers to encoding information by controlling the realization, rather than the distribution, of policy randomness. We show that PRAM preserves the prescribed action distribution for every state and every message, while still enabling reliable information transmission through the MDP dynamics. Under mild conditions, the posterior probability of the transmitted message converges to one as the trajectory length increases, ensuring asymptotically reliable message recovery. We further demonstrate that PRAM remains effective under model mismatch, where the encoder and decoder rely on an estimated transition kernel rather than the true system dynamics, making the method particularly relevant to learning-based systems.

Our main contributions are summarized as follows. \textbf{(i)} We introduce the idea of embedding information in policy randomness and formulate it as DPIC. \textbf{(ii)} We develop a posterior-matching randomness modulation scheme based on common randomness and prove its per-message distribution-preserving property. \textbf{(iii)} We establish the asymptotic reliability of PRAM. \textbf{(iv)} We demonstrate the effectiveness of  PRAM in multi-bit policy watermarking and trajectory provenance, including under model mismatch.

\section{Related Work}
%Our work connects to several lines of research spanning both applications and methodology. On the application side, watermarking is a direct use case of DPIC and has attracted growing research attention in recent years. On the methodological side, implicit communication through actions and feedback coding provide the communication-theoretic foundation for our formulation and algorithm.

\paragraph{Policy watermarking.} Existing policy-watermarking methods mainly focus on ownership verification through behavioral signatures. Early work embeds watermarks through designated state-action patterns~\citep{behzadan2019sequential}, while temporal watermarking exploits sequential behaviors for more robust verification~\citep{chen2021temporal}. More recent approaches extend behavioral watermarking to general agents~\citep{huang2025agent} and remotely observable robotic policies~\citep{amanda2026remotely}. These methods are primarily designed for zero-bit watermark detection~\citep{kuditipudi2023robust,christ2024undetectable}, where the detector can only distinguish between the presence and absence of a watermark. In contrast, multi-bit policy watermarking remains largely unexplored. A closely related work on multi-bit behavioral watermarking for LLM-based agents is TRACE~\citep{gao2026trace}, which also exploits controlled action sampling to embed information. However, its detector relies on direct observation of the executed actions, which is generally unavailable in many MDP and control settings. We therefore consider a more realistic setting in which the observer only observes the resulting state trajectory, and the embedded information must be conveyed through the stochastic system dynamics.

\paragraph{Sampling-based watermarking for generative models.} Compared with policy watermarking, multi-bit watermarking for generative models, particularly LLMs, has received substantially more attention. Existing methods embed information by modifying or controlling the token-sampling process while attempting to preserve text quality and statistical properties of the original model~\citep{kirchenbauer2023watermark,yoo2024advancing,wang2024towards,park2026watermod,zhu2026quantilemark}. However, by the nature of autoregressive LLM generation, these methods assume that the generated tokens—the actions under a sequential decision interpretation—are directly observable to the detector. As in~\citep{gao2026trace}, this leads to a fundamentally different decoding setting from ours and prevents these methods from being applied to general policy watermarking problems.

\paragraph{Steganography.}
Another line of research related to our problem is generative steganography, which embeds messages into generated content while concealing their presence. \cite{ziegler2019neural} employ arithmetic coding for linguistic steganography, while Meteor~\citep{kaptchuk2021meteor} provides cryptographic indistinguishability from model-generated content. Discop~\citep{ding2023discop} constructs distribution copies for secure embedding, and \cite{de2022perfectly} connect perfect security with minimum entropy coupling. Related techniques also enable distribution-preserving embedding in generated speech~\citep{chen2021distribution}. Autoregressive generation can be viewed as a special MDP in which the state is the generated prefix and the action is the next symbol, directly recoverable from the resulting prefix. In DPIC, however, the encoder selects actions while the decoder observes only the resulting states, generated through potentially noisy environment dynamics. These states need not uniquely determine the actions, and a desired coupling between messages and state trajectories need not yield an executable action policy. DPIC therefore introduces an additional channel-coding challenge: recovering information through environment dynamics while preserving the reference action distribution for every message.

\paragraph{Implicit communication and feedback coding.} The idea of DPIC is closely related to communicating via MDPs~\citep{sokota2022communicating,Karabay2019} and implicit communication in control systems~\citep{chen2025implicit,chen2026ICoCo,rips2025communication}. These works view the controlled system as a communication channel and mainly characterize the tradeoff between communication rate and control performance across different policies.
The most relevant work is Act2Comm~\citep{Act2Comm}, which studies communication through actions in finite-state MDPs, a setting similar to ours. 
Essentially, Act2Comm addresses the problem of channel coding over a finite-state channel with feedback and, like conventional channel coding, constrains the input distribution only after averaging over the message distribution. In contrast, DPIC imposes a stronger requirement: the prescribed action distribution must be preserved for every state and every message. This distinction is particularly important for applications such as policy watermarking, where the same identifier may be embedded repeatedly, and makes the DPIC problem fundamentally more restrictive.

Nevertheless, our method is inspired by feedback channel coding~\citep{Gallager-SKscheme,tatikonda2008capacity,permuter2009finite}, particularly posterior matching~\citep{shayevitz2011optimal} and its variant \citep{li2015efficient}. With feedback, the encoder observes the same channel outputs as the decoder and can therefore track the decoder's posterior belief over the message. Posterior matching exploits this shared posterior to adaptively determine what information should be conveyed next. Building on this principle, we introduce common randomness shared by the encoder and decoder and develop a posterior-matching-based scheme tailored to DPIC. The common randomness enables message-dependent action selection while preserving the prescribed action distribution for every message, which is the key distinction from conventional schemes.

\paragraph{Notations.} We use $s^t$ to denote the sequence $(s_1,s_2,\ldots,s_t)$. When it is necessary to distinguish a random variable from its realization, we use boldface notation, e.g., $\mathbf{m}$, for the random variable and the corresponding regular font, e.g., $m$, for its realization. 
%For any set $\set{A}$, let $\Delta(\set{A})$ denote the set of probability distributions over $\set{A}$.

\section{Problem Statement}
\subsection{System Model}

Consider a system in which an \textit{agent} interacts with its environment by taking actions. We model the system as an MDP defined by a tuple $(\mathcal{S},\mathcal{A},r,T)$, where $\set{S}$ is a finite state space, $\set{A}$ is a finite action space, $T:\set{S}\times \set{A}\times \set{S}\to [0,1]$ is the transition kernel, and $r:\set{S}\times \set{A}\to \mathbb{R}$ is the reward function. For simplicity, we index the actions by letting $\set{A}=\{1,2,\cdots,|\set{A}|\}$. At each time $t$, the agent observes the system state $s_t$ and selects an action $a_t$. The environment then transitions to a new state $s_{t+1}$ according to the distribution $T(\cdot|s_t,a_t)$. We further assume the presence of a \textit{passive observer} that can observe the state $s_t$ but cannot take actions or otherwise influence the environment. The system model is illustrated in Fig.~\ref{fig:system}.

\begin{wrapfigure}{r}{0.5\textwidth}
	\centering
	\includegraphics[width=0.8\linewidth]{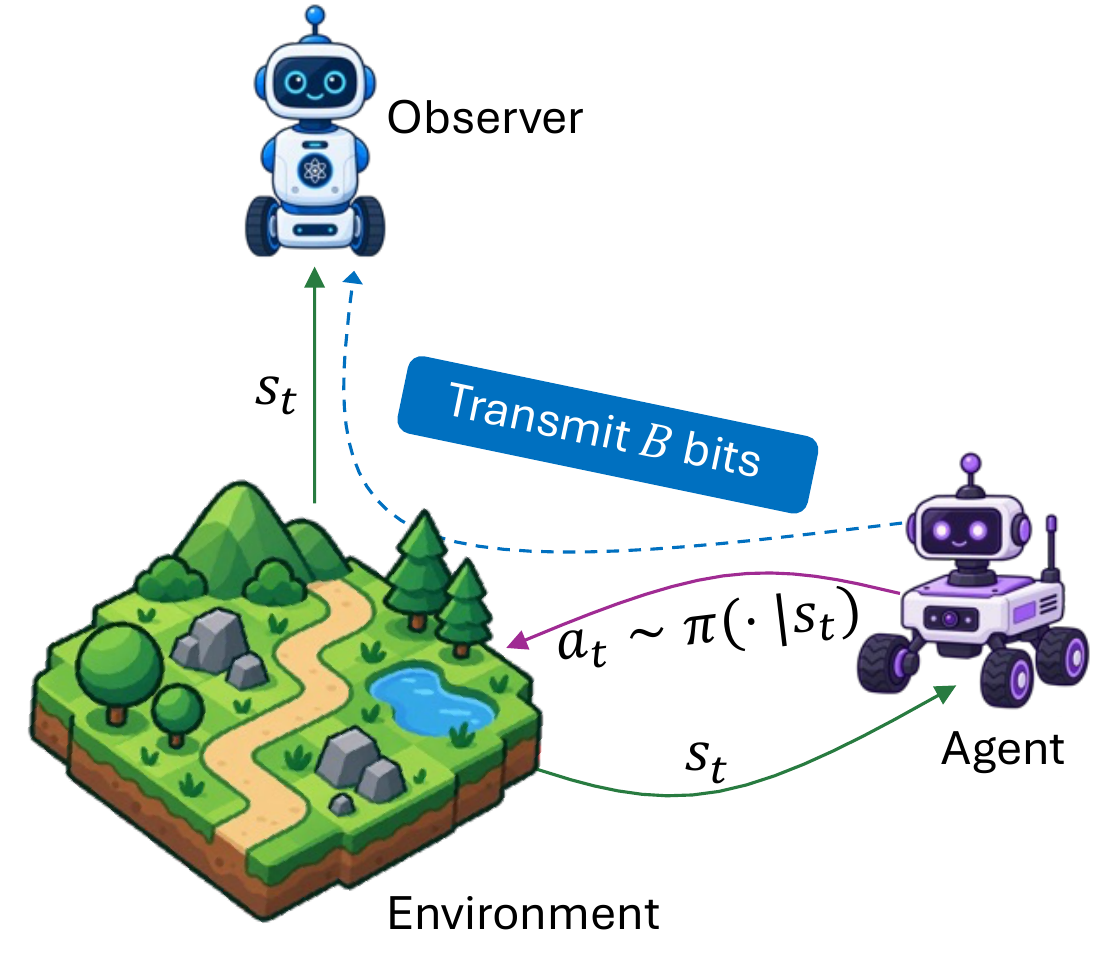}
	\caption{System model.}
	\label{fig:system}
\end{wrapfigure}

A stationary control policy $\pi:\set{S}\to  \Delta(\set{A})$ is a mapping that specifies an action distribution for each state, where $\Delta(\set{A})$ denote the probability simplex over $\set{A}$. We use $\pi(a|s)$ to denote the probability that policy $\pi$ selects action $a\in \set{A}$ in state $s\in \set{S}$.
Unlike standard MDP or RL problems, our objective is not to identify an optimal policy that maximizes the cumulative or average reward. Instead, we assume that the agent is already equipped with a suitable stationary policy $\pi$, referred to as the \textit{reference policy}. Our central question is: \textit{Can the agent encode a message in its actions while following $\pi$, such that the observer can accurately decode the message from the resulting sequence of observed states?}

In particular, let $\rvm$ denote a message drawn from a finite message set $\set{M}=\{0,1,\cdots,2^B-1\}$ according to a distribution $P_{\rvm}$, where $B$ denotes the message length in bits. Conditioned on $\rvm$, the agent selects an action sequence $a_{1:n}$, with each action $a_t$ determined using the message $\rvm$ and the state history up to the current time step $t$. After observing the resulting state sequence $s_{1:n+1}$, the observer produces an estimate $\hat{\rvm}$ of the message. The estimation error probability, also referred to as the decoding error probability, is then defined as $P^{(n)}_e \triangleq P(\hat{\rvm} \neq \rvm)$.
 
Our objective is twofold: \textbf{(i)} to minimize $P_e^{(n)}$ for a fixed horizon $n$, and \textbf{(ii)} to ensure, for every realization of $\rvm$, that the agent’s action at each time step remains distributed according to the reference policy, i.e., $\rva_t\sim \pi(\cdot|\rvs_t), \forall t$. This model has applications in a range of practical settings, including multi-bit policy watermarking and trajectory provenance, as will be demonstrated in greater detail in the experimental section.

 \subsection{Problem Formulation}
 In essence, our task is to enable the agent to transmit a message to the observer without relying on an external communication channel, instead using the MDP environment itself as the channel. This problem can be studied through the lens of implicit communication \citep{chen2025implicit}. 
 Specifically, the MDP environment can be interpreted as a finite-state channel (FSC) between the agent (encoder) and the observer (decoder): at each time step $t$, the action $a_t$ serves as the channel input, and the current state $s_t$ serves as the channel state. Upon execution of $a_t$, the environment transitions to a new state $s_{t+1}$, which constitutes the corresponding channel output. Under this interpretation, the reference policy $\pi$ specifies the desired input distribution of the FSC: conditional on the current channel state $s_t$, the channel input should be distributed according to $\pi(\cdot| s_t)$. 
 
 Within the framework of implicit communication, the problem reduces to designing a channel code for the FSC induced by the MDP. In channel coding, such a code consists of an encoder and a decoder. The encoder is a sequence of mappings $\set{E}=\{\set{E}_t:1\leq t\leq n\}$, where $	\set{E}_t:\set{M}\times\set{S}^t\times\set{A}^{t-1}\to\set{A}$ determines the action $a_t$ from the message, the state history, and the past actions.
 After observing the initial channel state $s_1$ and the subsequent $n$ channel outputs $s_{2:n+1}$, the observer applies the decoder mapping $\set{D}:\set{S}^{n+1}\to\set{M}$ to estimate the transmitted message. Here, $n$ denotes the \textit{blocklength}, corresponding to the number of time steps used to transmit the message.
 
 Conventional channel coding requires the channel input to follow a prescribed distribution, such as a capacity-achieving distribution, only after averaging over all messages. Our setting, however, imposes a stronger constraint: for every realization $m$ of $\rvm$, the channel input at time $t$ must follow the prescribed distribution $\pi(\cdot| s_t)$. At first glance, this requirement appears to preclude communication. Indeed, if $P(a_t|s_t,m)=\pi(a_t|s_t)$ for every $m$, then $\rva_t$ and $\rvm$ are conditionally independent given $s_t$, and hence the conditional mutual information $I(\rva_t;\rvm| s_t)=0$. As a result, input $\rva_t$ conveys no additional information about the message beyond that already contained in the current state.
 
 To address this challenge, we introduce common randomness shared between the agent and the observer. Specifically, let $\{\rvu_t\in \set{U}:t\ge 1\}$ be a sequence of independent and identically distributed (i.i.d.) random variables, with each $\rvu_t$ drawn from a distribution $P_{\rvu}$. At each time $t$, the realization of $\rvu_t$ is revealed to both the agent and the observer. We assume that $\rvu_t$ is independent of the message $\rvm$, the state history $\rvs^{t}$, and the action history $\rva^{t-1}$. In practice, such common randomness can be generated using a shared random seed derived from a secret key.
 
With the introduction of common randomness, we redefine the encoder as
 \begin{align*}
 	\set{E}_t:\set{M}\times\set{S}^t\times\set{A}^{t-1}\times \set{U}^t \to\set{A}, \ t\ge 1.
 \end{align*}
 Thus, at time $t$, the agent uses the shared random variables $\rvu^t$, together with the message and the available state and action histories, to determine the action $\rva_t$. Correspondingly, the decoder is defined as $\set{D}: \set{S}^{n+1}\times \set{U}^n \to\set{M}$.
The common randomness makes communication possible because the encoder can satisfy $I(\rva_t;\rvm|s_t)=0$ while maintaining $I(\rva_t;\rvm|s_t,u_t)>0$. In other words, although the action reveals no information about the message when conditioned only on the current state, it can depend strongly on the message once the shared randomness is known. Our objective is to design an encoder–decoder pair $(\set{E},\set{D})$ that enables efficient message transmission while ensuring that the agent adheres to the reference policy $\pi$.

\begin{figure}[t]
	\centering
	\includegraphics[width=0.8\linewidth]{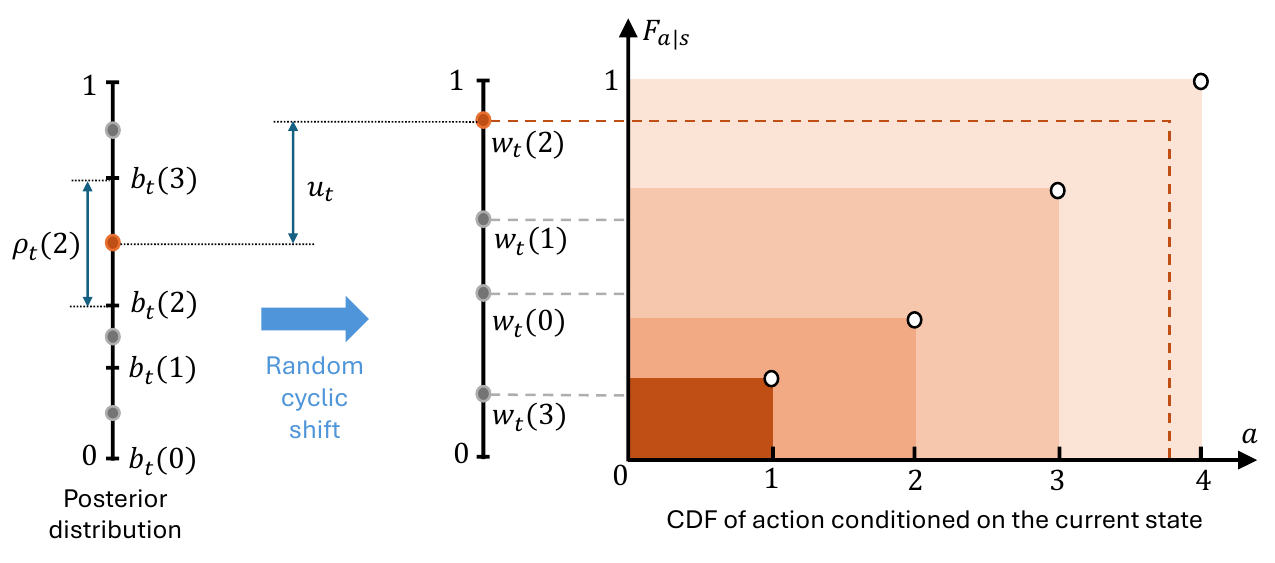}
	\caption{An illustration of the coding algorithm. Suppose $|\set{M}|=4$ and the message is $m^*=2$, then $a_t(m^*) = 4$ in this example.}
	\label{fig:algo}
\end{figure}

\section{Distribution-preserving Implicit Communication} \label{sec:PRAM}
In this section, we introduce a coding scheme that preserves the conditional action distribution $\rva_t\sim\pi(\cdot|\rvs_t)$ for every message $m$. We also provide a theoretical guarantee that, under this scheme, the observer’s decoding error probability vanishes as the blocklength $n$ tends to infinity.

\subsection{Coding Scheme}
Our method is largely inspired by posterior matching \citep{shayevitz2011optimal} and its variant \citep{li2015efficient}, both of which are coding schemes for feedback communication. However, these schemes preserve the desired input distribution only after averaging over all messages, rather than for each message individually. To overcome this limitation, we introduce the common randomness $\rvu_t$ and develop a novel coding scheme that leverages it to preserve the prescribed input distribution $\pi(\cdot|\rvs_t)$ for every message.

We first adapt the posterior matching principles introduced by \citet{shayevitz2011optimal} to incorporate common randomness. At each time $t$, the observer has access to the state sequence $\rvs^{t}$ and the common random sequence $\rvu^{t-1}$, which together provide information about the message $\rvm$.
Let  $\rho_t(\cdot)$ denote the observer's posterior distribution over $\rvm$, defined as
\begin{align*}
	\rho_t(m) \triangleq P(\rvm = m|\rvs^t,\rvu^{t-1}),\ \forall m\in \set{M}.
\end{align*}
Specifically, $\rho_1(m) = P_\rvm(m)$ as $\rvm$ and $\rvs_1$ are independent. As long as $\rho_t$ is not a point mass, the observer retains some uncertainty about $\rvm$ and therefore requires additional information to identify it. The posterior matching principle addresses the following question: \textit{given $\rvs^t$ and $\rvu^{t-1}$, what additional information is needed to determine $\rvm$?} A natural candidate is a random variable $\rvw_t$ satisfying the following conditions: (i) $\rvw_t$ is independent of $(\rvs^t,\rvu^{t-1})$; and (ii) the message $\rvm$ can be uniquely determined by $(\rvw_t, \rvs^{t}, \rvu^{t})$. They are referred to as two \textit{design principles}.

Intuitively, the first principle ensures that $\rvw_t$ does not convey any ``old information'' already available to the observer, while the second guarantees that $(\rvw_t,\rvu_t)$ carries sufficient ``new information" to determine the message. Given an appropriate construction of $\rvw_t$, the action $\rva_t$ is then selected as a function of $\rvw_t$ while ensuring that its conditional distribution matches the reference policy $\pi(\cdot\mid\rvs_t)$.

Following this idea, we are now ready to introduce our coding scheme. First, define
\begin{align*}
	b_t(m) \triangleq P(\rvm<m|\rvs^t,\rvu^{t-1}) = \sum_{j<m} \rho_t(j).
\end{align*}
Note that $(\rvs^t,\rvu^{t-1})$ is also available to the agent; therefore, the agent can also compute $\rho_t$ and $b_t$ at each time.
We then construct the information point at time $t$ as follows:
\begin{align*}
	\rvw_t(\rvm) = \left[b_t(\rvm) + l_t(\rvm)\rho_t(\rvm) + \rvu_t\right] \mod 1,
\end{align*}
where $l_t: \set{M}\to (0,1)$ is a deterministic function and $\rvu_t\sim \text{Unif}[0,1]$ independent of all the other variables. For each $m\in\set{M}$, the quantity $b_t(m)+l_t(m)\rho_t(m)$ lies within the interval $(b_t(m),b_t(m+1))$, and $\rvw_t(m)$ is obtained by applying a random cyclic shift to this point, as illustrated in Figure~\ref{fig:algo}. We write $\rvw_t$ in place of $\rvw_t(\rvm)$ when referring to the random variable induced by $\rvm$, rather than to its value for a particular realization of $\rvm$.
The choice of $l_t$ determines the pairwise distances between information points and thus affects the observer's decoding error probability. In this paper, we set $l_t(m)=0.5$ for all $m\in\set{M}$, as this choice simplifies the reliability analysis.

It is straightforward to verify that $\rvw_t$ satisfies the two design principles stated above. First, for every $m\in\set{M}$ and any realization $(s^t,u^{t-1})$, we have $\rvw_t(m)\sim\operatorname{Unif}[0,1]$,
which implies that $\rvw_t$ is independent of $(\rvs^t,\rvu^{t-1})$. Second, the message $\rvm$ can be uniquely recovered from $\rvw_t$, $\rvu_t$, $\rho_t$, and $b_t$. Because $\rho_t$ and $b_t$ are determined by $(\rvs^t,\rvu^{t-1})$, the second design principle also holds.

The next step is to select an action based on $\rvw_t$ while preserving the desired distribution. To this end, let $F_{\rva_t|\rvs_t}(\cdot)$ denote the cumulative distribution function (CDF) of $\pi(\cdot|\rvs_t)$. Its generalized inverse is defined as
\begin{align*}
	 F^{-1}_{\rva_t|\rvs_t}(w) := \inf \{a\in \set{A}: F_{\rva_t|\rvs_t}(a)\ge w \}, \forall w\in [0,1].
\end{align*}
The action at time $t$ is then determined by $\rva_t = F^{-1}_{\rva_t|\rvs_t}(\rvw_t) $.
Upon execution of $\rva_t$, the MDP transitions to a new state $\rvs_{t+1}$, which is observed by both the agent and the observer. Given any realization $s^{t+1}$, the posterior distribution is then updated according to Bayes' rule:
\begin{align} \label{eq:rho-update}
	\rho_{t+1}(m) = \frac{\rho_t(m) T(s_{t+1}|s_t,a_t(m))}{\sum_j \rho_t(j) T(s_{t+1}|s_t,a_t(j))}, \ \forall m\in \set{M},
\end{align}
where $a_t(m) = F^{-1}_{\rva_t|\rvs_t}(\rvw_t(m))$. At the end of the transmission horizon, i.e., time $n+1$, the observer can decode the message using the posterior distribution $\rho_{n+1}$. A natural choice is the maximum a posteriori (MAP) estimator:
\begin{align*}
	\hat{\rvm} = \argmax_{m\in \set{M}} \rho_{n+1}(m).
\end{align*}
We refer to this method as \textit{posterior-matching randomness modulation (PRAM)}. 

It remains to show that PRAM preserves the prescribed input distribution for every message $m$; this is formally established in the following result.

\begin{theorem}[Distribution Preservation]\label{thm: DP}
	Under the PRAM scheme, the following holds for every realization $m\in \set{M}$ of the message $\rvm$: 
	\begin{align*}
		P(a_t|s^t,u^{t-1},m) = \pi(a_t|s_t), \ \forall a_t, s^t.
	\end{align*}
\end{theorem}
\begin{proof}
	The proof is provided in Appendix~\ref{app:thm1}.
\end{proof}

Note that PRAM selects each action $\rva_t$ based on $\rvs^{t},\rvu^{t}$, and $\rvm$.
Hence Theorem~\ref{thm: DP} means that, when averaged over the common random variable $\rvu_{t}$, PRAM preserves the action distribution prescribed by the reference policy. Consequently, the communication capability is achieved without degrading the expected control reward. This distribution-preserving property is a key advantage of PRAM. Our analysis therefore focuses on distribution preservation and reliable communication through controlled Markov dynamics, which are not addressed in their setting.

PRAM is inspired by the randomized posterior-matching construction of \citep{li2015efficient}, and in particular adopts their idea of applying a common random cyclic shift to posterior intervals. Our use of this common randomness, however, serves a different purpose. Rather than optimizing the rate or error exponent of conventional feedback communication, PRAM uses common randomness to satisfy the per-message distribution-preserving constraint of DPIC.

\subsection{Reliability Analysis}
According to the design principles of PRAM, $\rvw_t$ carries information about the message $\rvm$, and so does the action $\rva_t$, since it is generated as a function of $\rvw_t$. However, this information is conveyed to the observer only through the stochastic state transitions of the system and may therefore be corrupted by transition noise. It is thus not immediately clear whether PRAM enables reliable message recovery. In this subsection, we establish a theoretical guarantee of asymptotic reliability: the observer's decoding error probability can be made arbitrarily small by increasing the blocklength.

We first introduce the following assumptions, which are essential for our analysis.
\begin{assumption}  \label{as: recurrence}
	The Markov chain induced by the reference policy $\pi$ is irreducible.
\end{assumption}
\begin{assumption}\label{as: KL}
	Under the reference policy $\pi$, there exists a state $s^*$ such that $\pi(\cdot|s^*)$ is nondegenerate and $D_{\text{KL}}(T(\cdot|s^*,a)|| T(\cdot|s^*,a'))>0$ for any pair of distinct actions $a$ and $a'$ in the support of $\pi(\cdot|s^*)$.
\end{assumption}
Irreducibility is a standard assumption in the analysis of finite-state MDPs and naturally yields the property required in Assumption~\ref{as: recurrence}. Assumption~\ref{as: KL} ensures the existence of at least one recurrent state at which the reference policy is genuinely stochastic and the transition distributions induced by different actions are distinguishable. Now we are ready to present the reliability guarantee.

\begin{theorem}[Asymptotic Reliability] \label{thm:reliability}
	Let $m^*\in \set{M}$ be the message transmitted by the agent. Then,  under Assumptions 1 and 2,
	\begin{align*}
		\rho_n(m^*)  \xrightarrow{\mathrm{a.s.}} 1 \quad \text{ as } n\to \infty.
	\end{align*}
As a result, the observer's decoding error probability $P^{(n)}_e\to 0$ as $n\to \infty$.
\end{theorem}
\begin{proof}
	The proof is provided in Appendix~\ref{app:thm2}.
\end{proof}
The significance of Theorem~\ref{thm:reliability} lies in the per-message distribution-preserving constraint of PRAM. As discussed earlier, requiring the action distribution to coincide with the reference policy for every message implies $I(\rva_t;\rvm|\rvs_t)=0$, meaning that, conditioned on the current state, the action alone reveals no additional information about the message. The introduction of common randomness $\rvu_t$, however, creates a message-dependent coupling between the action selection and the shared randomization, making communication possible without altering the prescribed action distribution. Theorem~\ref{thm:reliability} confirms that this mechanism is sufficient for reliable communication: despite the per-message distribution-preserving constraint, the observer can recover the transmitted message with vanishing error probability as the blocklength grows.

\section{Experiments} \label{sec:exp}
This section presents experimental results demonstrating the effectiveness of the proposed PRAM method. We evaluate PRAM in three MDP environments: ReFrozenLake with 64 states and 4 actions, ConTaxi with 500 states and 6 actions, and Stochastic Battery Management (SBM) with 1000 states and 5 actions. Detailed descriptions of these environments are provided in Appendix~\ref{app: envs}.

\begin{table}[t]
	\centering
	\caption{Total variation (TV) distance between the empirical action distribution induced by PRAM and the corresponding reference policy for different messages.}
	\label{tab:tv_distance}
	\begin{tabular}{llcccccccc}
		\toprule
		\multirow{2}{*}{\textbf{Environment}}
		& \multirow{2}{*}{\shortstack{\textbf{Reference}\\\textbf{Policy}}}
		& \multicolumn{8}{c}{\textbf{Message}} \\
		\cmidrule(lr){3-10}
		&
		& \textbf{0} & \textbf{1} & \textbf{2} & \textbf{3}
		& \textbf{4} & \textbf{5} & \textbf{6} & \textbf{7} \\
		\midrule
		
		\multirow{2}{*}{ReFrozenLake}
		& Policy 1
		& 0.012 & 0.012 & 0.010 & 0.011
		& 0.011 & 0.010 & 0.010 & 0.009 \\
		
		& Policy 2
		& 0.015 & 0.011 & 0.010 & 0.010
		& 0.011 & 0.011 & 0.012 & 0.011 \\

		\midrule
		
		\multirow{2}{*}{ConTaxi}
		& Policy 1
		& 0.007 & 0.006 & 0.006 & 0.007
		& 0.006 & 0.007 & 0.007 & 0.006 \\
		
		& Policy 2
		& 0.008 & 0.008 & 0.008 & 0.009
		& 0.008 & 0.008 & 0.008 & 0.008 \\

		\midrule
		
		\multirow{2}{*}{SBM}
		& Policy 1
		& 0.023 & 0.023 & 0.023 & 0.023
		& 0.023 & 0.023 & 0.023 & 0.023 \\
		
		& Policy 2
		& 0.029 & 0.028 & 0.028 & 0.029
		& 0.029 & 0.030 & 0.028 & 0.029 \\
		
		\bottomrule
	\end{tabular}
\end{table}

\paragraph{Distribution preservation.}
We first verify whether PRAM preserves the action distribution for every state and every message. To this end, we quantify the total variation (TV) distance between the reference policy $\pi$ and the empirical action distribution induced by PRAM under each message $m\in\set{M}$. Specifically, we define
$D_{\mathrm{TV}}(m)\triangleq \frac{1}{2} \sum_{s} \mu(s)  \sum_{a} |\pi(a|s) - P(a|s,m)|$,
where $\mu(\cdot)$ denotes the stationary state distribution induced by the reference policy, and $P(a|s,m)$ is the empirical probability of selecting action $a$ in state $s$ when message $m$ is transmitted using PRAM. We set the message length to $B=3$ bits and blocklength to $n=20$, and evaluate $D_{\mathrm{TV}}(m)$ for all eight possible messages. The results are reported in Table~\ref{tab:tv_distance}. The TV distance remains close to zero across all tested messages and environments, empirically confirming the per-message distribution-preserving property of PRAM. Consequently, PRAM can embed information without altering the nominal action distribution and hence without degrading the expected control reward. Additional results in Appendix~\ref{app:add-exp} further show that the control reward obtained under each transmitted message closely matches that of the reference policy.
\begin{figure}[t]
	\vspace{-5pt}
	\begin{subfigure}{0.33\linewidth}
		\centering
		\includegraphics[width=0.95\linewidth]{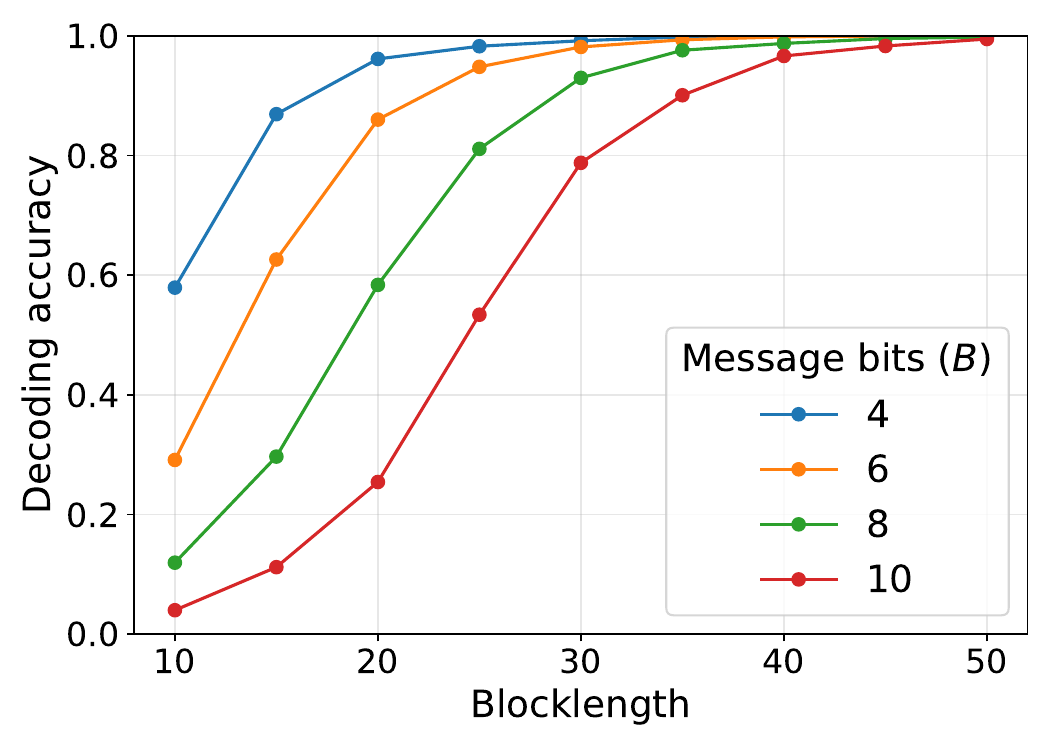}
		%\caption{\small{Trade-off between BER vs Reward over different rates.}}
		\caption{ReFrozenLake}
		\label{fig:accuracy_RFL}
	\end{subfigure} 
	\begin{subfigure}{0.33\linewidth}
		\centering
		\includegraphics[width=0.95\linewidth]{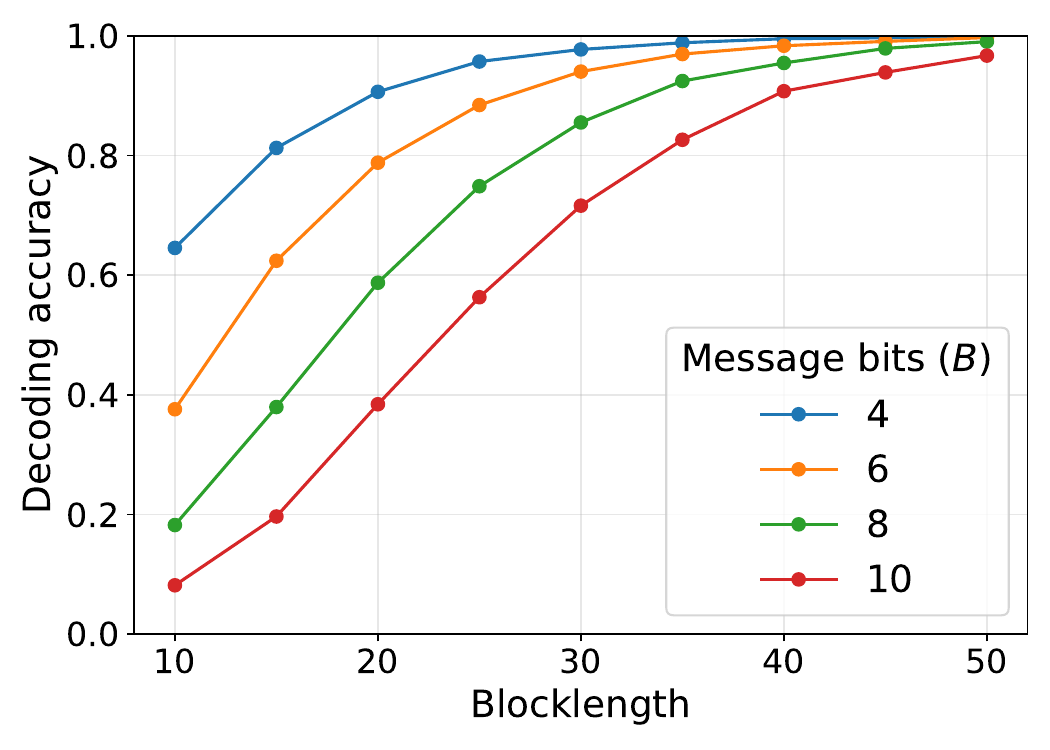}
		%\caption{\small{BER vs rate for a given reward}}
		\caption{ConTaxi}
		\label{fig:accuracy_Taxi}
	\end{subfigure} 
	\begin{subfigure}{0.33\linewidth}
		\centering
		\includegraphics[width=0.95\linewidth]{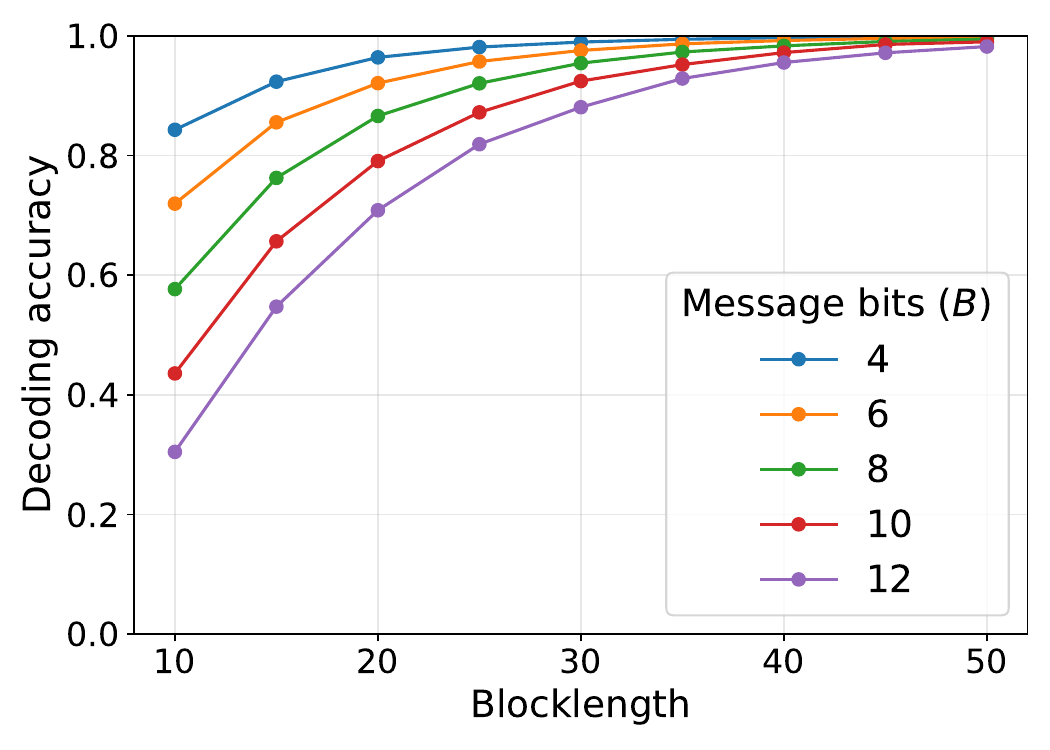}
		%\caption{\small{Reward vs rate for a given BER}}
		\caption{SBM}
		\label{fig:accuracy_SBM}
	\end{subfigure} 
	\caption{PRAM performance on policy watermarking in ReFrozenLake and SBM, and trajectory provenance in ConTaxi.}
	\label{fig:accuracy}
	\vspace{-5pt}
\end{figure}

\paragraph{Applications.}
The ability of PRAM to embed information into stochastic policy realizations while preserving the prescribed action distribution naturally enables several downstream applications. We consider two representative tasks: \textit{multi-bit policy watermarking} and \textit{trajectory provenance}. In policy watermarking, a fixed multi-bit identifier is associated with a policy and repeatedly embedded into its behavior, providing a persistent signature for ownership verification. In contrast, trajectory provenance aims to determine the origin of an observed trajectory, e.g., which deployed agent or execution session produced it; the embedded information may therefore vary across episodes, agents, or deployment instances. We evaluate multi-bit policy watermarking in ReFrozenLake and SBM, and trajectory provenance in ConTaxi. For each task, we vary the message length and report the decoding accuracy as a function of the trajectory blocklength.  The decoding accuracy is estimated from 4000 trials.

Although these two applications differ in the semantics of the embedded information, both can be realized through the same PRAM communication mechanism. For the watermarking task, we report performance averaged over different identifiers, making the experimental protocol comparable to that used for trajectory provenance. The results are shown in Figure~\ref{fig:accuracy}. The decoding accuracy is defined as the fraction of trials in which all message bits are correctly recovered, i.e., $1-P_e^{(n)}$. As shown in the figure, the decoding accuracy increases monotonically with the blocklength and approaches one for all tested message lengths even at reasonably low trajectory lengths ($\sim 40$-$50$).

\paragraph{Effect of randomness. }
We next investigate how different sources of randomness affect the communication performance of PRAM. In particular, we consider randomness arising from both the stochastic reference policy and the system dynamics. To study the effect of policy randomness, we obtain the reference policy using entropy-regularized value iteration with temperature parameter $\alpha$, where $\alpha=0$ corresponds to the optimal deterministic policy and a larger $\alpha$ induces a more stochastic policy. We then evaluate PRAM under reference policies with different temperatures in ReFrozenLake and SBM. The results are shown in Figures~\ref{fig:rand_RFL} and~\ref{fig:rand_SBM}. As expected, greater policy randomness generally leads to higher decoding accuracy, since it provides more degrees of freedom for embedding information.

We examine the effect of system randomness in ConTaxi, where the taxi moves in the intended direction with probability $p$ and in either lateral direction with probability $(1-p)/2$.  Hence the resulting transition uncertainty is characterized by the entropy $H(p)= -p\log p - (1-p) \log \frac{1-p}{2}$.
We fix the reference policy and evaluate PRAM in ConTaxi under different rainy probabilities $p$, as shown in Figure~\ref{fig:rand_Taxi}. From a communication perspective, randomness in the system dynamics acts as channel noise and makes the transmitted information more difficult to recover from the state trajectory. Therefore, lower transition uncertainty generally leads to higher decoding accuracy.

\begin{figure}[t]
	\vspace{-5pt}
	\begin{subfigure}{0.33\linewidth}
		\centering
		\includegraphics[width=0.95\linewidth]{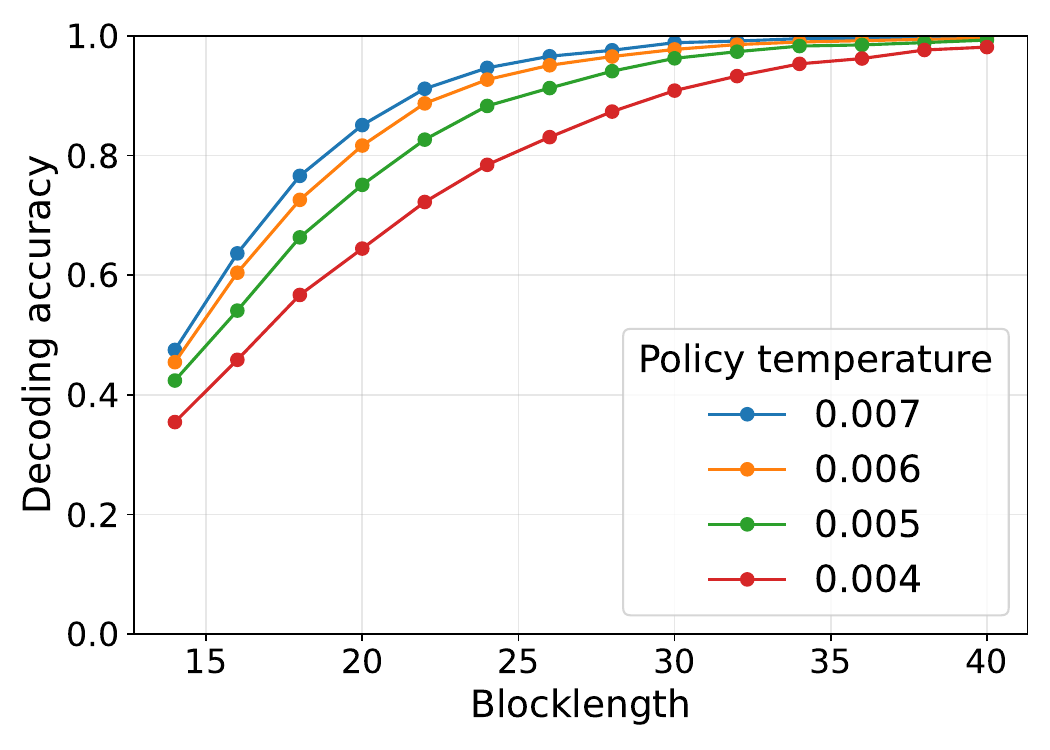}
		%\caption{\small{Trade-off between BER vs Reward over different rates.}}
		\caption{ReFrozenLake }
		\label{fig:rand_RFL}
	\end{subfigure} 
	\begin{subfigure}{0.33\linewidth}
		\centering
		\includegraphics[width=0.95\linewidth]{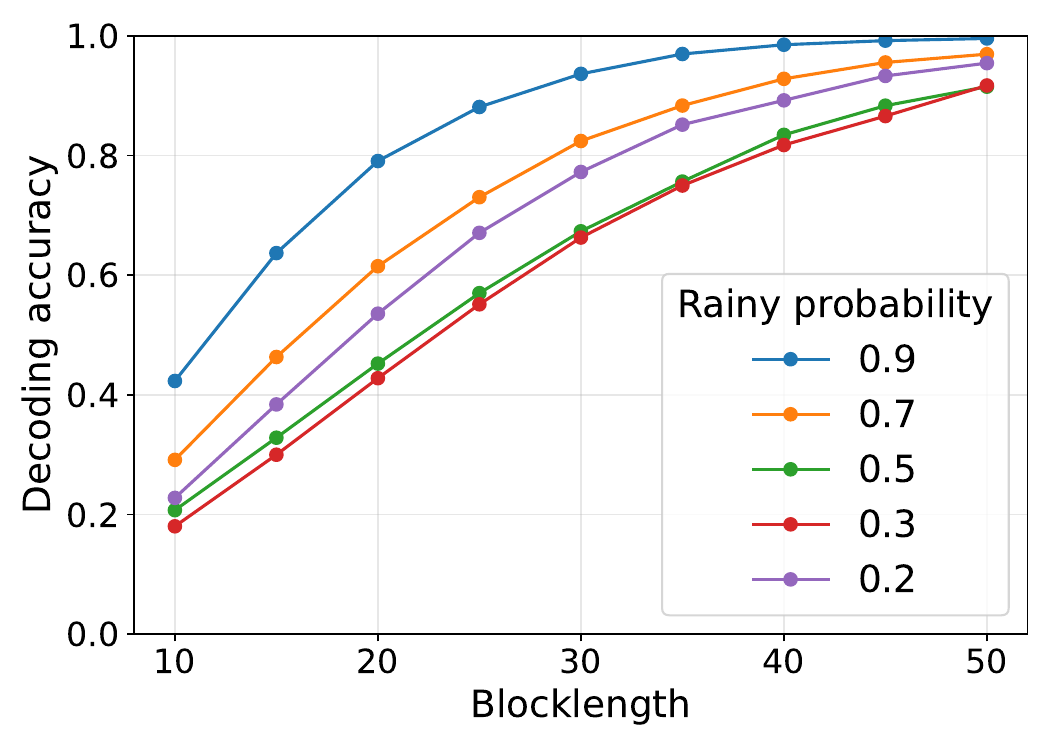}
		%\caption{\small{BER vs rate for a given reward}}
		\caption{ConTaxi}
		\label{fig:rand_Taxi}
	\end{subfigure} 
	\begin{subfigure}{0.33\linewidth}
		\centering
		\includegraphics[width=0.95\linewidth]{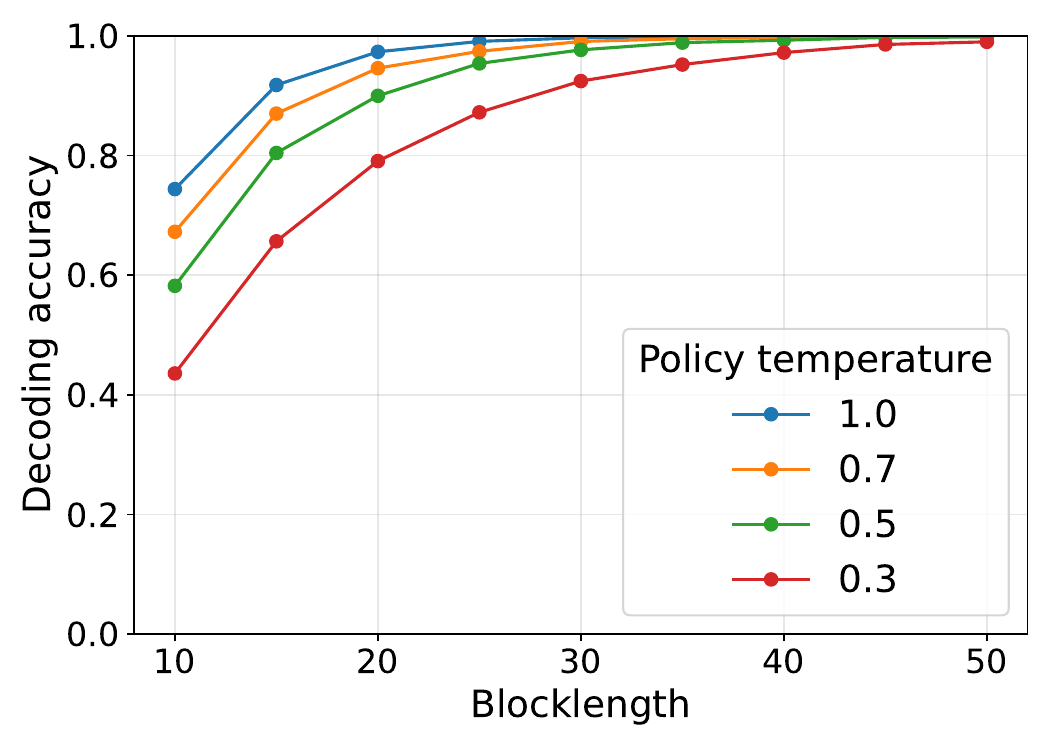}
		%\caption{\small{Reward vs rate for a given BER}}
		\caption{SBM }
		\label{fig:rand_SBM}
	\end{subfigure} 
	\caption{Effect of randomness on PRAM. (a) and (c): Policy randomness is varied through the temperature parameter. (b) System randomness is varied through the rainy probability.}
	\label{fig:rand}
	\vspace{-5pt}
\end{figure}

\begin{figure}[t]
	\vspace{-5pt}
	\begin{subfigure}{0.33\linewidth}
		\centering
		\includegraphics[width=0.95\linewidth]{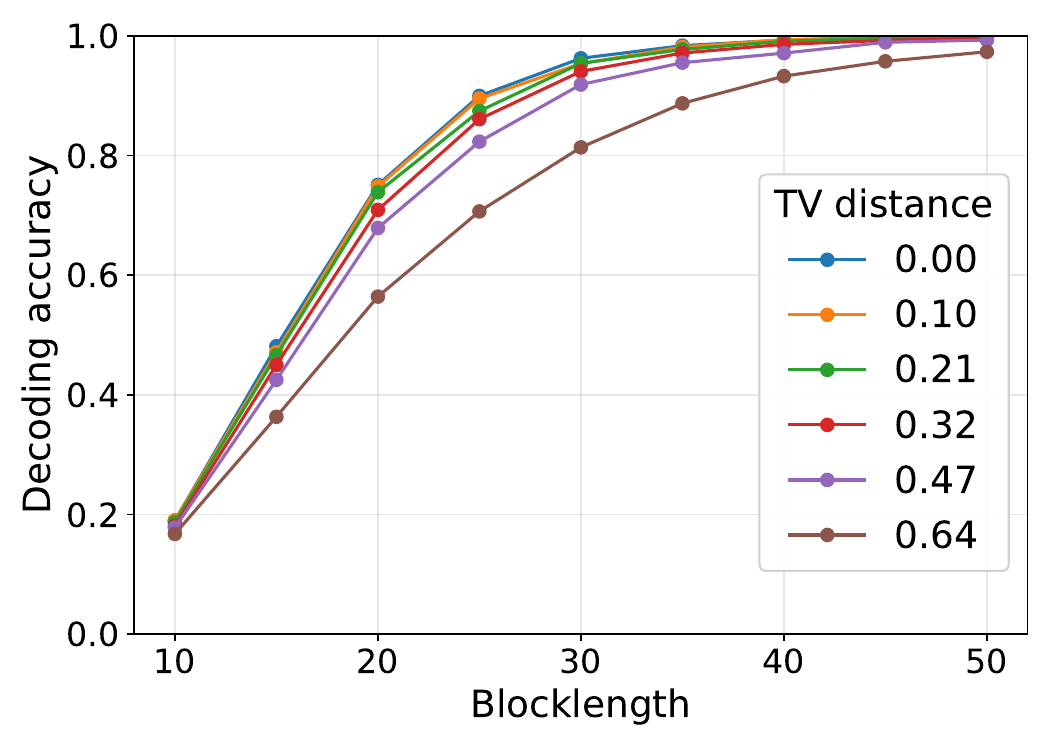}
		%\caption{\small{Trade-off between BER vs Reward over different rates.}}
		\caption{ReFrozenLake}
		\label{fig:mismatch-RFL}
	\end{subfigure} 
	\begin{subfigure}{0.33\linewidth}
		\centering
		\includegraphics[width=0.95\linewidth]{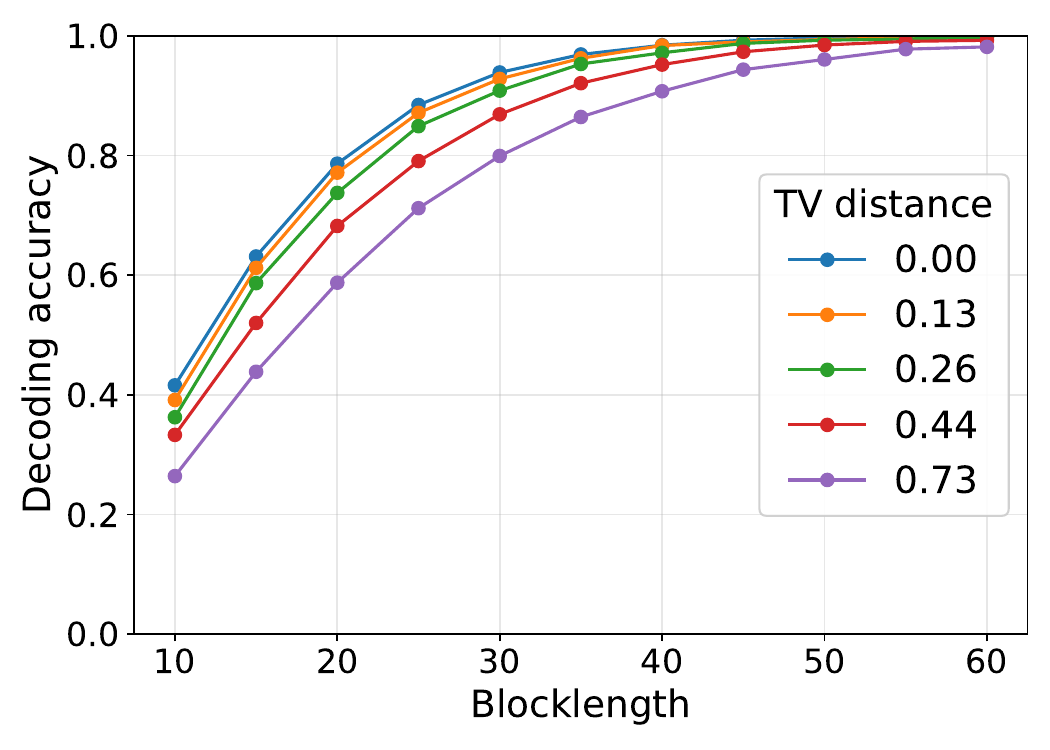}
		%\caption{\small{BER vs rate for a given reward}}
		\caption{ConTaxi}
		\label{fig:mismatch-Taxi}
	\end{subfigure} 
	\begin{subfigure}{0.33\linewidth}
		\centering
		\includegraphics[width=0.95\linewidth]{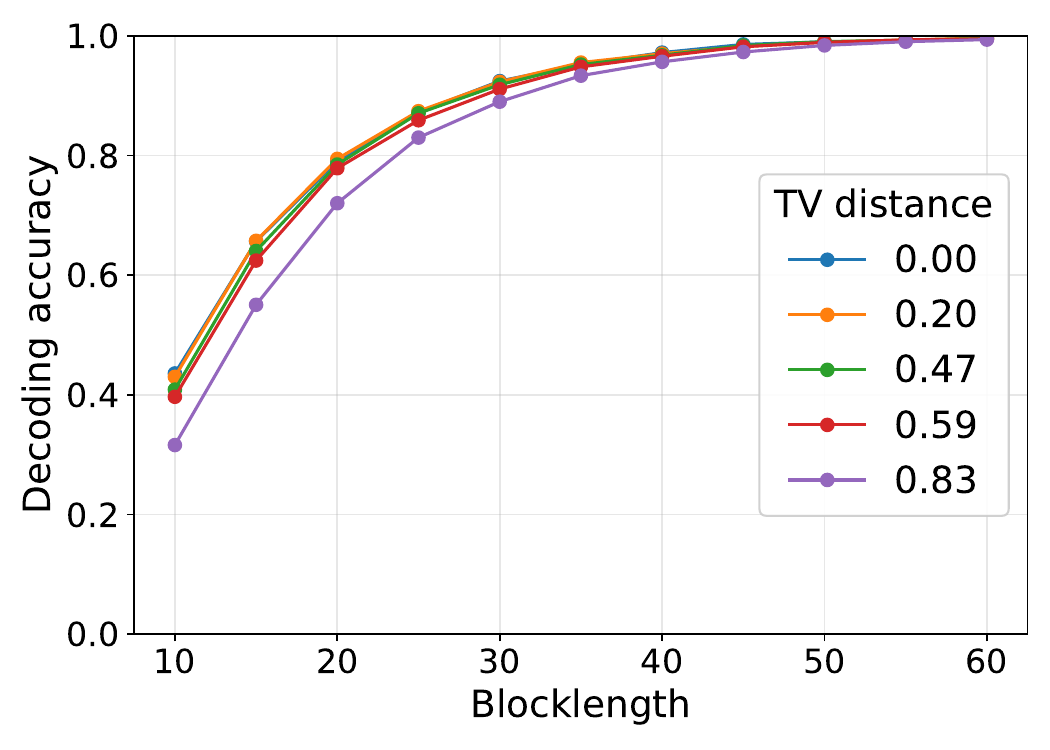}
		%\caption{\small{Reward vs rate for a given BER}}
		\caption{SBM}
		\label{fig:mismatch-SBM}
	\end{subfigure} 
	\caption{Performance of PRAM under model mismatch. The TV distance measures the discrepancy between the true transition kernel and the estimated kernel used by PRAM. }
	\label{fig:mismatch}
	\vspace{-5pt}
\end{figure}
\paragraph{Model mismatch.}
We further examine the robustness of PRAM to model mismatch. Specifically, suppose that neither the agent nor the observer knows the true transition kernel $T$ of the MDP, and both instead rely on an estimated kernel $\hat{T}\neq T$. PRAM therefore performs action selection and posterior updates using $\hat{T}$, while the observed state trajectory is generated according to the true kernel $T$. Such model mismatch naturally arises in RL, where the transition dynamics must typically be estimated from interaction data.

We quantify the discrepancy between $T$ and $\hat{T}$ using the occupancy-weighted TV distance $	d(T,\hat{T})\triangleq \frac{1}{2} \sum_{s} \mu(s) \sum_{a} \pi(a|s)  \sum_{s'} |T(s'|s,a) - \hat{T}(s'|s,a)|$,
where $\mu$ is the stationary state distribution induced by the reference policy $\pi$.
Figure~\ref{fig:mismatch} reports the decoding performance of PRAM under different levels of model mismatch. For each environment, we estimate $\hat{T}$ from interaction trajectories, with larger amounts of data yielding more accurate estimates and hence smaller TV distances. We then fix the message length and evaluate the decoding accuracy as a function of the blocklength. The results demonstrate that PRAM is robust to moderate model mismatch: high decoding accuracy can still be achieved even when $\hat{T}$ differs noticeably from the true transition kernel. Moreover, across the tested mismatch levels, the decoding accuracy continues to improve with the blocklength and approaches one for sufficiently long trajectories.

\section{Conclusion}
This paper studied how to embed information into an agent's behavior without altering the action distribution prescribed by a stochastic control policy. We introduced the idea of \emph{randomness modulation} and formulated the problem as DPIC. The key idea is to view the MDP environment as a communication channel, with the agent's action serving as the channel input and the resulting state as the channel output. Unlike conventional channel coding, DPIC imposes a stronger per-message distribution-preserving requirement. To address this challenge, we introduced common randomness shared by the encoder and decoder and developed the PRAM scheme. By modulating action realizations, PRAM establishes message-dependent couplings between the action and the common random variable, enabling information recovery from the observed state trajectory while preserving the prescribed action distribution. We proved the per-message distribution-preserving property of PRAM and established its asymptotic reliability. Extensive experiments demonstrated reliable message recovery, robustness to model mismatch, and the applicability of PRAM to tasks including multi-bit policy watermarking and trajectory provenance. Extending DPIC to partially observable settings, where the observer has only incomplete or noisy state information~\citep{krishnamurthy2016partially}, is an important and challenging direction for future work.

\clearpage

%\subsubsection*{Acknowledgments}

\bibliography{iclr2027_conference}
\bibliographystyle{iclr2027_conference}

\clearpage
\appendix
\section*{Appendix}
\section{Technical Proofs} \label{app:proof}
\subsection{Proof of Theorem~\ref{thm: DP}} \label{app:thm1}
\begin{proof}
	The desired result follows from an argument analogous to that used in Lemma II.1 of  \citep{shayevitz2011optimal}. In particular, for every $m\in \set{M}$,
	\begin{align*}
		P(F^{-1}_{\rva_t|\rvs_t}(\rvw_t(m))\le a) & = P\left( \inf\{a':  F_{\rva_t|\rvs_t}(a')\ge \rvw_t(m)\} \le a \right) \\
		& \overset{(a)}{=} P\left( F_{\rva_t|\rvs_t}(a) \ge \rvw_t(m)\right) \\
		& \overset{(b)}{=} F_{\rva_t|\rvs_t}(a).
	\end{align*}
	Here, (a) holds because a CDF is non-decreasing and continuous from the right, and (b) follows from the fact that $\rvw_t(m)\sim\operatorname{Unif}[0,1]$ for every $m\in \set{M}$. By definition, the above implies that $\rva_t = F^{-1}_{\rva_t|\rvs_t}(\rvw_t(m))\sim \pi(\cdot|\rvs_t)$ for  every realization $m$ of the message.
\end{proof}

\subsection{Proof of Theorem~\ref{thm:reliability}} \label{app:thm2}
This part provides the proof of Theorem~\ref{thm:reliability}. Before that, we need the following lemma.
\begin{lemma} \label{lem: Ktm}
	Let $m^*\in \set{M}$ be the message transmitted by the agent. For any $m\neq m^*$, define
	\begin{align*}
		K_t(m) = \mathbf{1}\{\rvs_t = s^*, \rva_t(m^*)\neq \rva_t(m) \}.
	\end{align*}
Then we have
{
\begin{align*}
	\liminf_{n\to \infty} \frac{1}{n} \sum_{t=1}^nK_t(m) > 0 \quad\text{a.s.}
\end{align*}}
\end{lemma}
\begin{proof}
	Define the message point before applying random cyclic shift as
	\begin{align*}
		w'_t(m) = b_t(m) + \frac{1}{2} \rho_t(m).
	\end{align*}
	Let
	\begin{align*}
		\alpha_t(m) \triangleq \min \{ |w'_t(m^*)-w'_t(m)|, 1- |w'_t(m^*)-w'_t(m)| \}
	\end{align*}
be the circular distance between $w'_t(m^*)$ and $w'_t(m)$.
Then for any time $t$, we have $\alpha_t(m) \le 0.5$ and
\begin{align*}
	\alpha_t(m) \ge \frac{\rho_t(m^*) + \rho_t(m)}{2} \ge \frac{\rho_t(m^*) }{2}. 
\end{align*}

For any $m\neq m^*$, define the likelihood ratio
\begin{align} \label{eq:Rtm}
	R_t(m) \triangleq \frac{P(\rvs^{t},\rvu^{t-1}|m)}{P(\rvs^{t},\rvu^{t-1}|m^*)}
\end{align}
According to Bayes' rule, we have
\begin{align*}
	\rho_t(m^*) &= \frac{P(\rvs^{t},\rvu^{t-1},m^*)}{P(\rvs^{t},\rvu^{t-1})} = \frac{P_\rvm(m^*) P(\rvs^{t},\rvu^{t-1}|m^*)}{\sum_{m\in \set{M}} P_\rvm(m) P(\rvs^{t},\rvu^{t-1}|m)} \\
	& = \frac{P_\rvm(m^*)}{P_\rvm(m^*) + \sum_{m\neq m^*} P_\rvm(m) R_t(m)}.
\end{align*}
We next establish a lower bound on $\rho_t(m^*)$ by analyzing $R_t(m)$. By definition,
\begin{align} \label{eq:Rtm-up}
	R_{t+1}(m) = R_t(m) \frac{T(\rvs_{t+1}|\rvs_t,\rva_t(m))}{T(\rvs_{t+1}|\rvs_t,\rva_t(m^*))}.
\end{align}
Then we can verify that $R_t(m)$ is a nonnegative supermartingale. To see this, let $\set{G}_t=\sigma(\rvs^t,\rvu^{t-1})$ denote the sigma-algebra generated by $(\rvs^t,\rvu^{t-1})$. Similarly, define $\mathcal{H}_t = \sigma(s^t,u^t)$. 
Conditional on $\mathcal H_t$, both candidate actions are known,
and the next state $s_{t+1}$ has distribution $T(\cdot|s_t,a_t(m^*))$. Let $\mathcal{S}_+:= \{s\in \mathcal{S}:T(\rvs_{t+1}=s|s_t,a_t(m^*))>0\}$.
Then
\begin{align*}
	\mathbb{E}[R_{t+1}(m)|\set{H}_t] & = R_t(m) \sum_{s\in \mathcal{S}_+} T(\rvs_{t+1}=s|s_t,a_t(m^*)) \frac{T(\rvs_{t+1}=s|s_t,a_t(m))}{T(\rvs_{t+1}=s|\rvs_t,\rva_t(m^*))} \\
	& \le R_t(m).
\end{align*}
Since \(\mathcal G_t\subseteq\mathcal H_t\), the tower property gives
\[
\begin{aligned}
\mathbb E[R_{t+1}(m)\mid\mathcal G_t]
&=
\mathbb E\!\left[
\mathbb E[R_{t+1}(m)\mid\mathcal H_t]
\,\middle|\,\mathcal G_t
\right]\\
&\le
\mathbb E[R_t(m)\mid\mathcal G_t]\\
&=R_t(m).
\end{aligned}
\]
Applying the nonnegative supermartingale convergence theorem yields
\begin{align*}
	R_t(m) \to R_\infty(m) < \infty \qquad \text{a.s.}
\end{align*}
Therefore, there exist constants $C_m >0$ and $N_m>1$ such that
\begin{align} \label{eq: alpha-tm}
	\alpha_t(m) \ge C_m, \quad \forall t\ge N_m.
\end{align}

At the state $s^*$ specified in Assumption~\ref{as: KL}, define
\begin{align*}
	f(u) = F^{-1}_{\rva|s^*}(u) 
\end{align*}
Recall that
\begin{align*}
\rva_t(m) = f(\rvu_t + w'_t(m)\mod 1) 
\end{align*}
Since $\rvu_t\sim \text{Unif}[0,1]$, we can shift coordinates and define
\begin{align*}
	q(\alpha) \triangleq P_\rvu(f(\rvu)\neq f(\rvu + \alpha)\mod 1).
\end{align*}
Here, $q(\alpha)$ is the probability that $\rva_t(m^*) \neq \rva_t(m) $ given that the circular distance between the two message points is $\alpha$. Since $\pi(\cdot|s^*)$ is non-deterministic, we have $q(\alpha)>0$ for every $\alpha\in (0,1)$. Moreover, $q(\alpha)$ is continuous in $\alpha$ and 
\begin{align*}
	\alpha_t(m)\in [C_m, 0.5],
\end{align*}
compactness gives
\begin{align*}
	q_m \triangleq \inf_{\alpha \in [C_m, 0.5]} q(\alpha) > 0.
\end{align*}
Therefore,
\begin{align} \label{eq:Pa-neq}
	P(\rva_t(m^*) \neq \rva_t(m) |\mathcal{G}_t) \ge q_m >0
\end{align}
for sufficiently large $t$. Define the state visitation frequency under the coding scheme as
\begin{align*}
	\theta_n(s) = \frac{1}{n} \sum_{t=1}^{n} P(\rvs_t = s).
\end{align*}
{
By Theorem~\ref{thm: DP}, under PRAM the state process is a Markov chain with
transition kernel $P_\pi(s'|s)=\sum_a \pi(a|s)T(s'|s,a)$. 
Then Assumption~\ref{as: recurrence} and the finite-state
Markov chain ergodic theorem imply that  $\theta(s^*)\triangleq \lim_{n\to \infty} \theta_n(s^*)$ exists and
\begin{align} \label{eq:theta}
\frac1n\sum_{t=1}^n\mathbf1{\{s_t=s^*\}}
\longrightarrow\theta(s^*)>0
\quad\text{a.s.}
\end{align}
Now, independence and uniformity of \(u_t\) give
\begin{align*}
	\mathbb{E}[K_t(m)|\set{G}_t] = \mathbf{1}\{\rvs_t = s^*\} q(\alpha_t(m)).
\end{align*}
By the preceding lower bound on $q(\alpha_t(m))$, almost surely,
for all sufficiently large $t$, $\mathbb E[K_t(m)\mid\mathcal G_t]\ge q_m\mathbf1\{\rvs_t=s^*\}$
where $q_m>0$ may depend on the sample path.
}
Now define the martingale difference
\begin{align*}
	Z_t(m) = K_t(m) - \mathbb{E}[K_t(m)|\set{G}_t].
\end{align*}
Since $K_t(m)\in \{0,1\}$, $\mathbb{E}[K_t(m)|\set{G}_t]\in [0,1]$. We thus have $|Z_t(m)|\le1$. The martingale strong law~\citep{hall1980martingale} gives
\begin{align*}
	\frac{1}{n} \sum_{t=1}^{n} Z_t(m) \to 0 \quad \text{a.s.}
\end{align*}
{
This implies that
\begin{align}
	\liminf_{n\to \infty} \frac{1}{n} \sum_{t=1}^n K_t(m) = \liminf_{n\to \infty} \frac{1}{n} \sum_{t=1}^n \mathbb{E}[K_t(m)|\mathcal{G}_t] \ge q_m \theta(s^*)>0 \quad\text{a.s.}
\end{align}}
This completes the proof.

\end{proof}

Now we are ready to establish Theorem~\ref{thm:reliability}.
{
\begin{proof}[Proof of Theorem~\ref{thm:reliability}]
Let $m^*$ denote the transmitted message and fix an arbitrary message
$m\ne m^*$. 
%Throughout this proof, probabilities and expectations are taken under $\mathbb P_{m^*}$. 
Define $\mathcal H_t\triangleq\sigma(s^t,u^t)$ as the sigma-algebra generated by $s^t$ and $u^t$. Similarly, define $\mathcal{G}_t\triangleq\sigma(s^t,u^{t-1})$.
Let $R_t(m)$ be the likelihood ratio defined in \eqref{eq:Rtm} in Lemma~1, and set
\[
L_t=\sqrt{R_t(m)}.
\]
For brevity, write
\[
P_t(s')=T(s'|s_t,a_t(m^*)),
\qquad
Q_t(s')=T(s'|s_t,a_t(m)),
\]
and define their affinity by
\[
V_t\triangleq\sum_{s'\in\mathcal S}\sqrt{P_t(s')Q_t(s')}.
\]

Conditional on $\mathcal H_t$, both candidate actions are known
and $s_{t+1}$ has distribution $P_t$. According to the update of likelihood-ratio $R_t$ given in \eqref{eq:Rtm-up}, we have
\begin{align*}
    L_{t+1} = L_t \sqrt{\frac{Q_t(s_{t+1})}{P_t(s_{t+1})}}.
\end{align*}
It then follows that
\begin{align}\label{eq:affinity-update}
\mathbb E_{m^*}[L_{t+1}\mid\mathcal H_t]
&=
L_t\sum_{s':P_t(s')>0}
P_t(s')\sqrt{\frac{Q_t(s')}{P_t(s')}} =L_tV_t,
\end{align}
where $\mathbb{E}_{m^*}[\cdot]$ denotes expectation under the probability law induced by transmitting message $m^*$.
This identity remains valid when transition probabilities vanish,
since terms with $P_t(s')=0$ contribute zero to $V_t$.
By the Cauchy--Schwarz inequality, $V_t\le1$. As we have already shown in Lemma~\ref{lem: Ktm}, taking conditional
expectations with respect to $\mathcal G_t$ shows that $\{L_t\}$
is a nonnegative supermartingale with respect to
$\{\mathcal G_t\}$. Hence $L_t$ converges almost surely to a finite
limit $L_\infty$.

Taking expectations in~\eqref{eq:affinity-update} and summing,
we obtain
\[
\mathbb E_{m^*}\!\left[
\sum_{t=1}^n L_t(1-V_t)
\right]
=
\mathbb E_{m^*}[L_1]-\mathbb E_{m^*}[L_{n+1}]
\le1,
\]
where $L_1=1$. 
Since $L_t\ge0$ and $V_t\le1$, the partial sums
\[
S_n:=\sum_{t=1}^n L_t(1-V_t)
\]
are nonnegative and nondecreasing. Hence, by the monotone
convergence theorem,
\[
\mathbb E_{m^*}\!\left[
\sum_{t=1}^{\infty}L_t(1-V_t)
\right]
=
\lim_{n\to\infty}
\mathbb E_{m^*}\!\left[
\sum_{t=1}^{n}L_t(1-V_t)
\right]
\le1.
\]
The infinite sum is therefore a nonnegative random variable
with finite expectation, so
\[
\sum_{t=1}^{\infty}L_t(1-V_t)<\infty
\quad\text{a.s.}
\]

Let $\mathcal A_*=\operatorname{supp}\pi(\cdot|s^*)$ denote the set of actions assigned positive probability by the reference policy $\pi$ at state $s^*$, where
$s^*$ is the state in Assumption~\ref{as: KL}, and define
\[
\eta=
\min_{\substack{a,a'\in\mathcal A_*\\a\ne a'}}
\left(
1-\sum_{s'\in\mathcal S}
\sqrt{T(s'|s^*,a)T(s'|s^*,a')}
\right).
\]
Assumption~2 ensures that the transition distributions for
distinct actions in $\mathcal A_*$ are distinct. Their affinities
are therefore strictly less than one. Since $\mathcal A_*$
is finite, $\eta>0$.

Recalling
\[
K_t(m)=
\mathbf1\{\rvs_t=s^*,\,\rva_t(m^*)\ne \rva_t(m)\}.
\]
If $K_t(m)=1$, then $\rvs_t=s^*$ and the two candidate actions
are distinct elements of $\mathcal A_*$ almost surely.
The definition of $\eta$ therefore gives $1-V_t\ge\eta$.
If $K_t(m)=0$, we still have $1-V_t\ge0$.
Thus, in both cases,
we have $1-V_t\ge\eta K_t(m)$ almost surely. Consequently,
\[
\sum_{t=1}^{\infty}L_tK_t(m)<\infty
\quad\text{a.s.}
\]
By Lemma~1,
$\sum_{t=1}^{\infty}K_t(m)=\infty$ almost surely.
Since $L_t\to L_\infty$, these two conclusions imply
$L_\infty=0$ almost surely. Thus
\[
R_t(m)\longrightarrow0
\quad\text{a.s.}
\]

The message set is finite, so this convergence holds simultaneously
for every $m\ne m^*$. Bayes' rule now gives
\[
\rho_t(m^*)=
\frac{P_m(m^*)}
{P_m(m^*)+\sum_{m\ne m^*}P_m(m)R_t(m)}
\longrightarrow1
\quad\text{a.s.}
\]
The MAP decoder therefore selects $m^*$ eventually almost surely.
Bounded convergence, followed by averaging over the finite message
set, yields $P_e^{(n)}\to0$.
\end{proof}
}

\section{Experiment Details} \label{app:exp}
\begin{figure}[t]
	\vspace{-5pt}
	\begin{subfigure}{0.33\linewidth}
		\centering
		\includegraphics[width=0.95\linewidth]{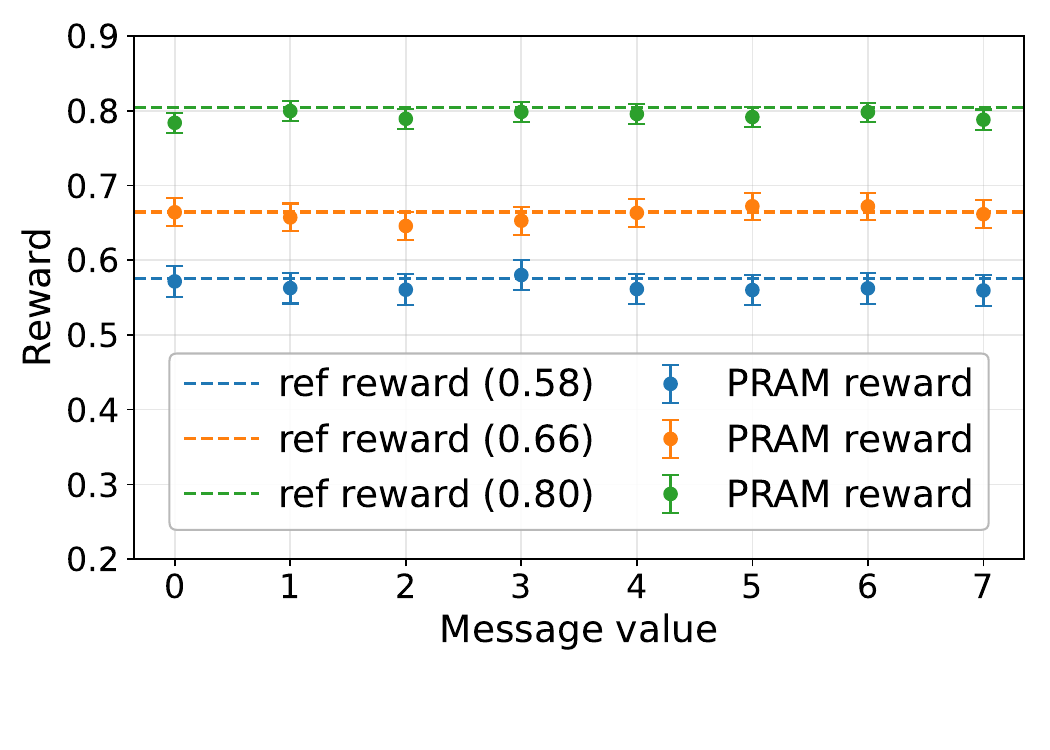}
		%\caption{\small{Trade-off between BER vs Reward over different rates.}}
		\caption{ReFrozenLake}
		\label{fig:reward-RFL}
	\end{subfigure} 
	\begin{subfigure}{0.33\linewidth}
		\centering
		\includegraphics[width=0.95\linewidth]{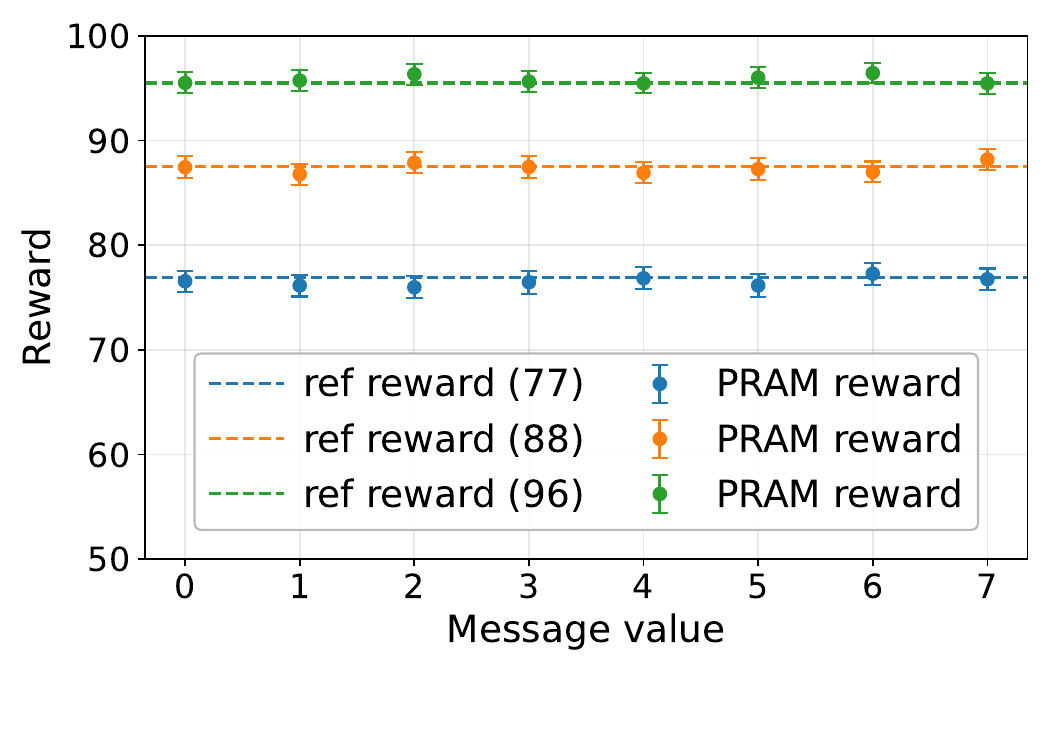}
		%\caption{\small{BER vs rate for a given reward}}
		\caption{ConTaxi}
		\label{fig:reward-Taxi}
	\end{subfigure} 
	\begin{subfigure}{0.33\linewidth}
		\centering
		\includegraphics[width=0.95\linewidth]{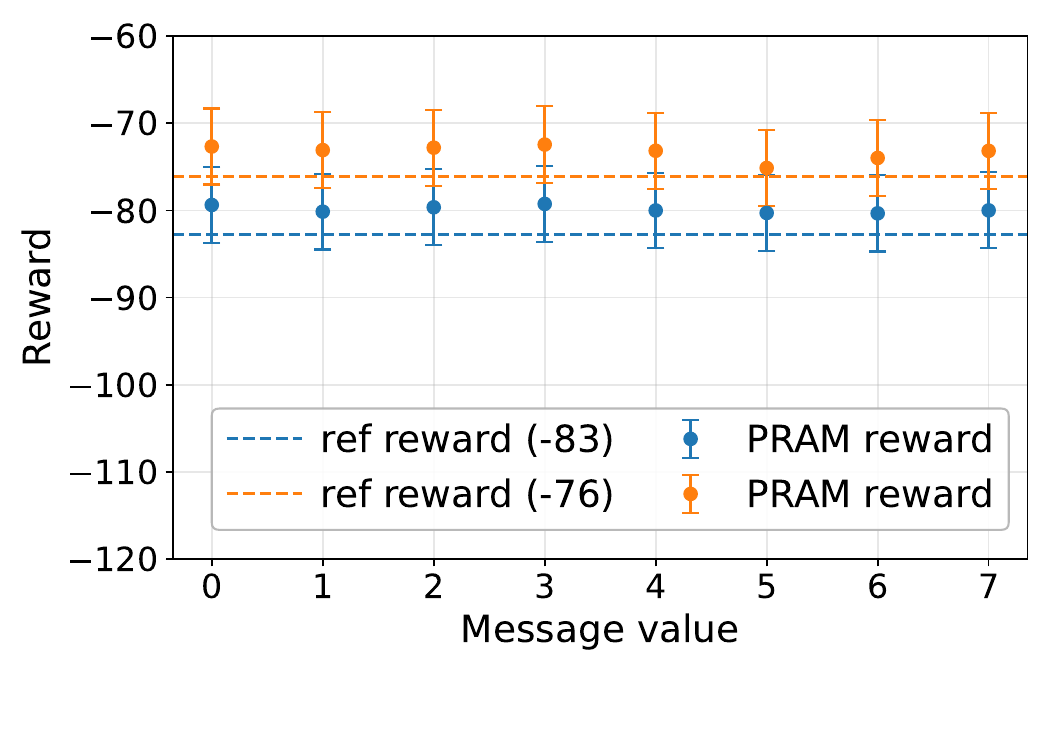}
		%\caption{\small{Reward vs rate for a given BER}}
		\caption{SBM}
		\label{fig:reward-SBM}
	\end{subfigure} 
	\caption{Control reward of transmitting different messages using PRAM. }
	\label{fig:reward}
	\vspace{-5pt}
\end{figure}

\subsection{Experimental environments} \label{app: envs}
This section provides brief descriptions of the experimental environments used to evaluate the proposed PRAM method. Detailed Python implementations of these environments are included in the supplementary material and will be open-sourced after the review process.

\paragraph{ReFrozenLake.} The Revised FrozenLake (ReFrozenLake) is a variant of FrozenLake in Gymnasium, modified to provide a continuing finite-state MDP suitable for long-horizon communication experiments. The environment retains the original grid layout and slippery transition dynamics, but removes terminal states before reaching the goal. In particular, reaching a hole no longer terminates the episode; instead, the agent incurs a penalty and is randomly transitioned to another valid position in the grid. Reaching the goal yields a positive reward and terminates the trajectory. 

\paragraph{ConTaxi.} The Continuing Taxi (ConTaxi) is a variant of Taxi-v4 in Gymnasium that transforms the original episodic task into a continuing environment. It preserves the original state and action spaces, including the taxi’s location, passenger location, destination, and six available actions. When the taxi successfully drops off a passenger, the episode does not terminate; instead, the taxi remains at the drop-off location and a new passenger origin and destination are sampled uniformly from all valid pairs. The environment therefore supports an uninterrupted sequence of trips while retaining Taxi-v4’s reward structure and optional rainy-mode dynamics.

\begin{figure}
    \centering
    \includegraphics[width=\linewidth]{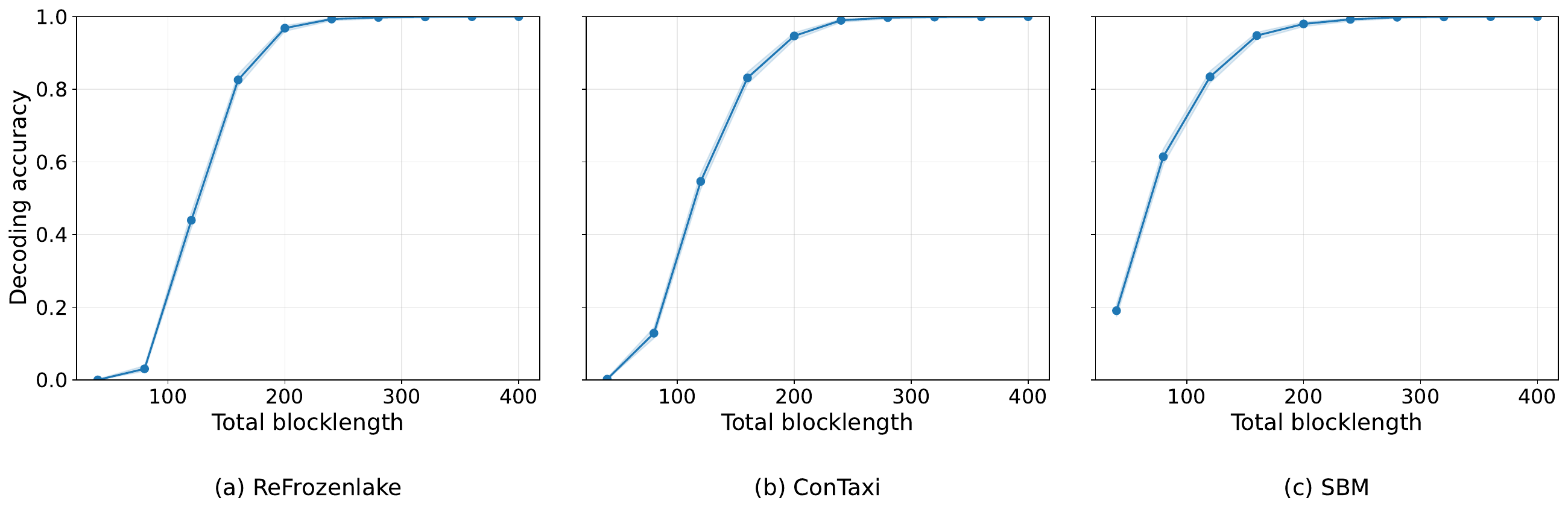}
    \caption{PRAM for 32-bit messages transmitted as four sequential 8-bit blocks.}
    \label{fig:32bit}
\end{figure}

\begin{figure}
    \centering
    \includegraphics[width=\linewidth]{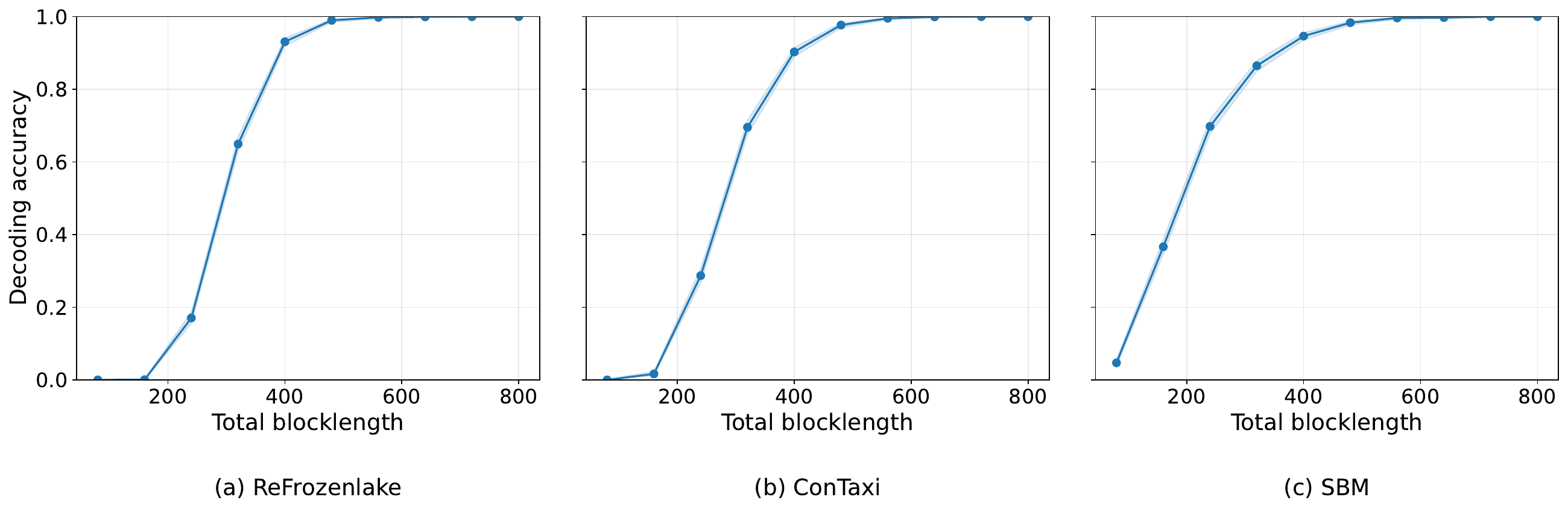}
    \caption{PRAM for 64-bit messages transmitted as four sequential 8-bit blocks.}
    \label{fig:64bit}
\end{figure}

\paragraph{SBM.}
Stochastic Battery Management (SBM) is a customized finite-state, finite-action energy-management environment designed to model stochastic battery charging and discharging. The state is $s=(b,d,p)$, where $b\in\{0,\ldots,19\}$ denotes the battery level, $d\in\{-4,\ldots,5\}$ denotes the net electricity demand, and $p\in\{1,\ldots,5\}$ denotes the electricity-price level, resulting in $1000$ states in total. The action space is $\mathcal{A}=\{-2,-1,0,1,2\}$, corresponding to different charging and discharging rates.

The commanded battery action is executed with stochastic perturbation: with probability $0.8$ it is executed as intended, while with probability $0.1$ it is shifted by $-1$ or $+1$, respectively, subject to the action and battery-capacity constraints. The net demand and electricity price independently evolve according to nearest-neighbor Markov chains. The overall transition kernel therefore factorizes into the battery, demand, and price transition probabilities.

The reward reflects the electricity purchasing cost, revenue from selling excess energy, and battery degradation cost. Let $x=b'-b$ denote the realized battery change and $g=d+x$ the resulting grid electricity exchange. The instantaneous reward is
\[
r=-p[g]^+ + \beta p[-g]^+ - \lambda |x|,
\]
where $[z]^+=\max\{z,0\}$, $\beta=0.5$ is the selling-price ratio, and $\lambda=0.1$ controls the battery degradation penalty. The environment is continuing and contains no terminal states. In our experiments, we set each episode to be 130 steps.

\subsection{Additional experimental results} \label{app:add-exp}

In DPIC, the per-message distribution-preserving requirement is motivated by the need to maintain control performance while embedding watermarks or other information into the agent's behavior. Table~\ref{tab:tv_distance} has shown that the empirical action distributions induced by PRAM closely match the corresponding reference policies. Here, we further verify that PRAM indeed preserves the control reward. Specifically, in each of the three environments, we select two or three reference policies and evaluate their average rewards, shown as dashed lines in the figure. For each reference policy, we then evaluate the reward achieved by PRAM when transmitting different messages. As shown in the figure, the rewards under PRAM are nearly identical across messages and closely match those of the corresponding reference policies.

{
We also conducted more experiments to evaluate the performance of transmitting longer bit streams. To avoid maintaining a belief over an exponentially large message alphabet, we divided each stream into consecutive 8-bit blocks and transmitted them sequentially using PRAM. Specifically, 32-bit and 64-bit messages were divided into four and eight blocks, respectively. Each block was decoded separately, and the belief was reset before transmitting the next block while the environment state continued to evolve. This approach requires only 256 candidate messages per block, regardless of the total stream length. We evaluated whole-message decoding accuracy: a transmission was counted as successful only when every block was decoded correctly.

Experiments were conducted in ReFrozenlake, ConTaxi, and SBM, with whole-message decoding accuracy estimated from 2,000 trials per setting. The blocklength per 8-bit block ranged from 10 to 100 in increments of 10, corresponding to total blocklengths of 40–400 for 32-bit messages and 80–800 for 64-bit messages. The results demonstrate that sequential block transmission enables PRAM to scale to longer bit streams while achieving high decoding accuracy and maintaining fixed belief-storage requirements and computational complexity per transmission step.

\subsection{Ablation Study}

 We constructed a baseline by removing the encoder's posterior updates from PRAM. In this variant, the encoder uses fixed message points without adapting them to state feedback. This ablation isolates the contribution of feedback adaptation through posterior updates.

As shown in Figure~\ref{fig:baseline}, for 4-bit messages, the baseline is able to achieve the same decoding accuracy as PRAM but requires a longer blocklength. 
As the message length increases, however, its accuracy remains low even at large blocklengths, despite improving as the blocklength increases. For example, for 10-bit messages, the baseline achieves approximately 0.25 accuracy at a blocklength of 100, whereas PRAM achieves 0.99 accuracy at a blocklength of only 50. These results demonstrate the benefit of adapting to feedback through posterior updates, a key component of PRAM's feedback coding mechanism.

\begin{figure}
    \centering
    \includegraphics[width=\linewidth]{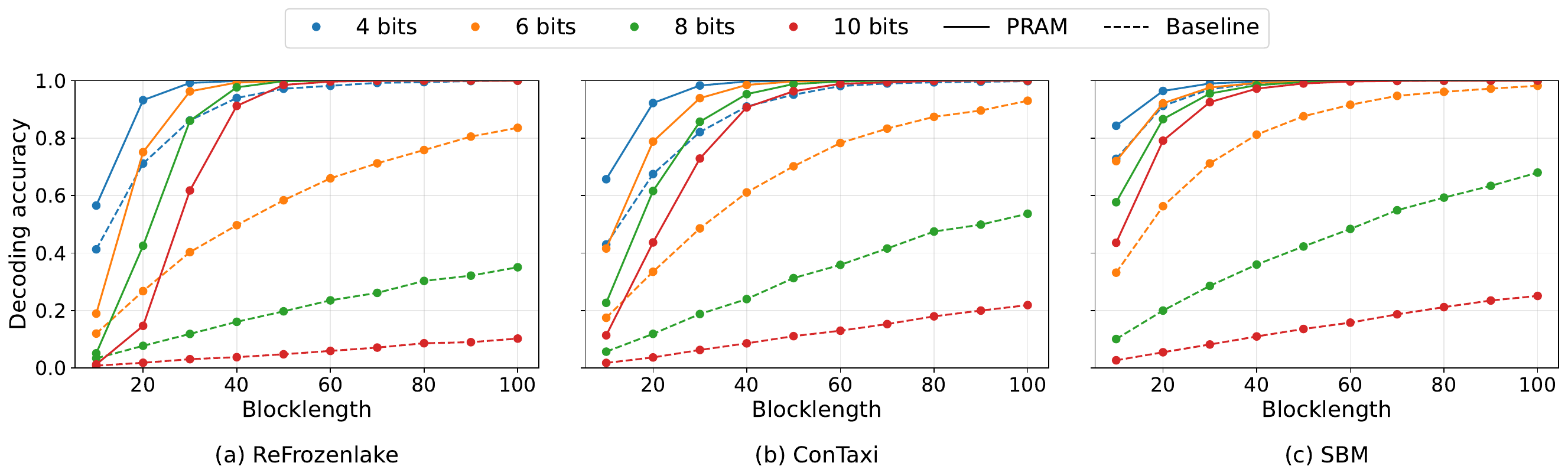}
    \caption{Performance comparison between PRAM and fixed-message-point baseline. This baseline is constructed by removing the encoder's posterior update in PRAM.}
    \label{fig:baseline}
\end{figure}

\subsection{Further Details of Model Mismatch Experiments}
This section provides more details about experiments underlying Figure~\ref{fig:mismatch}. The purpose of this group of experiments is to show that PRAM is robust to model mismatch, that is, it still achieves good performance when it relies on an estimated transition kernel $\hat{T}\neq T$. In the experiments, we execute the reference policy $\pi$ to collect trajectories from the agent's interactions with the environment. After each episode, we extract the transition tuples $(s,a,s')$, consisting of the current state, action, and next state, and add them to a dataset $\mathcal{D}$. Let $\kappa(s,a,s')$ denote the number of occurrences of transition $(s,a,s')$ in $\mathcal{D}$, and let $\kappa(s,a)=\sum_{s'}\kappa(s,a,s')$ denote the total number of transitions associated with state-action pair $(s,a)$. 
We then estimate the transition kernel using the following 
\begin{align}
    \hat{T}(s'|s,a) = \frac{\kappa(s,a,s') + \kappa_0}{\kappa(s,a) + \kappa_0 |\mathcal{S}|},
\end{align}
where $\kappa_0$ is a smoothing constant, set to $1$ in our experiments.

In Figure~\ref{fig:mismatch}, each curve $l$ corresponds to a transition kernel $\hat{T}$ estimated from a dataset $\mathcal{D}_l$. We quantify the estimation error using the previously defined TV distance between $T$ and $\hat{T}$, which generally decreases as the dataset grows. Tables~2-4 report the number of data-collection episodes and the corresponding estimation errors.

\begin{table}[H]
\centering
\caption{ReFrozenLake}
\begin{tabular}{|c|c|c|c|c|c|}
\hline
\rowcolor[HTML]{ECF4FF} 
Num. of episodes  & 20   & 50   & 100  & 200 &500 \\ \hline
TV distance & 0.64 & 0.47 & 0.32 & 0.21 &0.1 \\ \hline
\end{tabular}

\end{table}

\begin{table}[H]
\centering
\caption{ConTaxi}
\begin{tabular}{|c|c|c|c|c|}
\hline
\rowcolor[HTML]{ECF4FF} 
Num. of episodes  & 200   & 1000   & 3000  & 10000 \\ \hline
TV distance & 0.73 & 0.44 & 0.26 & 0.13 \\ \hline
\end{tabular}

\end{table}

\begin{table}[H]
\centering
\caption{SBM}
\begin{tabular}{|c|c|c|c|c|}
\hline
\rowcolor[HTML]{ECF4FF} 
Num. of episodes  & 1000   & 5000   & 10000  & 50000 \\ \hline
TV distance & 0.83 & 0.59 & 0.47 & 0.2 \\ \hline
\end{tabular}

\end{table}

}

\section{Algorithm Pseudocode}  \label{app:algo}
The encoding procedure of PRAM at the agent is summarized in Algorithm~\ref{alg:1}.
\begin{algorithm}[H]
\caption{Posterior-matching randomness modulation (PRAM): Encoder}
\label{alg:1}
\begin{algorithmic}[1]

\REQUIRE Transition kernel $T$, blocklength $n$, message $m^*\in \set{M}$, message prior distribution $P_\rvm$

\STATE \textbf{Initialization:}
Let $\rho_1 = P_\rvm$. Observe the initial state $s_1$

\FOR{$t = 1$ to $n$}

    \STATE Compute for every $m$:
    \[
    b_t(m) \triangleq P(\rvm<m|\rvs^t,\rvu^{t-1}) = \sum_{j<m} \rho_t(j)
    \]

    \STATE Draw a sample of the common random variable $\rvu_t\sim\mathrm{Unif}[0,1]$

    \STATE Determine the information point:
    \[
    \rvw_t = \left[b_t(m^*) + l_t(m^*)\rho_t(m^*) + \rvu_t\right] \mod 1,
    \]

    \STATE Compute the action
    \[
        a_t = F^{-1}_{\rva_t|\rvs_t}(\rvw_t ) := \inf \{a\in \set{A}: F_{\rva_t|\rvs_t}(a)\ge \rvw_t  \}
    \]

    \STATE Apply action $a_t$ and observe the next state $s_{t+1}$. Then update the posterior
    \[
    \rho_{t+1}(m) = \frac{\rho_t(m) T(s_{t+1}|s_t,a_t(m))}{\sum_j \rho_t(j) T(s_{t+1}|s_t,a_t(j))}, \ \forall m\in \set{M}
    \]

\ENDFOR

\end{algorithmic}
\end{algorithm}

\end{document}